\documentclass[12pt]{article}      
\usepackage{natbib,epsfig,multirow,setspace,amsthm}

\AtEndOfClass{\RequirePackage{times}}
\usepackage{setspace}
\AtBeginDocument{%
  \addtolength\abovedisplayskip{0\baselineskip}%
  \addtolength\belowdisplayskip{0\baselineskip}%
}
\usepackage{mathtools}
\usepackage{amsthm,amsmath,amssymb,colonequals,multirow,graphicx,epstopdf}
\usepackage{caption}
\usepackage{subcaption}
\usepackage{color}
\newtheorem{prop}{Proposition}
\usepackage{rotating,booktabs}

\newcommand{\vecY}{\mathbf{X}}
\newcommand{\vecX}{\mathbf{Y}}
\newcommand{\vecZ}{\mathbf{Z}}
\newcommand{\ident}{\mathbf{I}}
\newcommand{\vecU}{\mathbf{U}}
\newcommand{\vecB}{\mathbf{B}}
\newcommand{\vecz}{\mathbf{z}}
\newcommand{\vecy}{\mathbf{x}}
\newcommand{\vecb}{\mathbf{b}}
\newcommand{\vecmu}{\mbox{\boldmath$\mu$}}
\newcommand{\vecalpha}{\mbox{\boldmath$\alpha$}}
\newcommand{\vecbeta}{\mbox{\boldmath$\beta$}}

\newcommand{\matsig}{\mbox{\boldmath$\Sigma$}}
\newcommand{\gam}{\mbox{\boldmath$\Gamma$}}
\newcommand{\del}{\mbox{\boldmath$\Delta$}}
\newcommand{\matTheta}{\mbox{\boldmath$\Theta$}}
\newcommand{\varthet}{\mbox{\boldmath$\vartheta$}}

\newcommand{\isum}{\sum_{i=1}^n}
\newcommand{\gsum}{\sum_{g=1}^G}
\newcommand{\inv}{^{\raisebox{.2ex}{$\scriptscriptstyle-1$}}}
\DeclareMathOperator{\E}{\mathbb{E}}
\DeclareMathOperator{\R}{\mathbb{R}}

\begin{document}

\title{Mixtures of Generalized Hyperbolic Distributions and Mixtures of Skew-t Distributions for Model-Based Clustering with Incomplete Data}

\author{Yuhong Wei, Yang Tang and Paul D.\ McNicholas }

\date{\small Dept.\ of Mathematics \& Statistics, McMaster University, Hamilton, Ontario, Canada.}

\maketitle

\begin{abstract}
Robust clustering from incomplete data is an important topic because, in many practical situations, real data sets are heavy-tailed, asymmetric, and/or have arbitrary patterns of missing observations. Flexible methods and algorithms for model-based clustering are presented via mixture of the generalized hyperbolic distributions and its limiting case, the mixture of  multivariate skew-t distributions. An analytically feasible EM algorithm is formulated for parameter estimation and imputation of missing values for mixture models employing missing at random mechanisms. The proposed methodologies are investigated through a simulation study with varying proportions of synthetic missing values and illustrated using a real dataset. Comparisons are made with those obtained from the traditional mixture of generalized hyperbolic distribution counterparts by filling in the missing data using the mean imputation method.\\[-10pt]

\end{abstract}

\section{Introduction}\label{sec:intro}

Finite mixture models are powerful and flexible tools for discovering unobserved heterogeneity in multivariate datasets. Assuming no prior knowledge of class labels, the application of finite mixture models in this way is known as model-based clustering. As \cite{mcnicholas16} points out, the association between mixture models and clustering goes back at least as far as \cite{tiedeman55}, who uses the former as a means of defining the latter. Gaussian mixture models are historically the most popular tool for model-based clustering and dominated the literature for quite some time~\citep[e.g.,][]{celeux95,fraley98,mclachlan03,bouveyron07,mcnicholas08,mcnicholas10}. The multivariate $t$-distribution, being a heavy-tailed alternative to the multivariate Gaussian distribution, made (robust) mixture modelling based on mixtures of multivariate $t$-distributions the most natural extension \citep[e.g.,][]{peel00,andrews11,andrews12,steane12,lin14b}. In many practical situations, however, real world datasets exhibit clusters that are not just heavy tailed but also asymmetric; furthermore, clusters can also be asymmetric yet not heavy tailed. Over the few past years, much attention has been paid to non-Gaussian approaches to model-based clustering and classification, including work on multivariate skew-$t$ distributions \citep[e.g.,][]{lin10,vrbik12,lee14,murray14a,murray14b,murray17b}, shifted asymmetric Laplace distributions \citep{franczak14}, multivariate power exponential distributions \citep{dang15}, multivariate normal inverse Gaussian distributions \citep{karlis09,ohagan16}, generalized hyperbolic distributions \citep{browne15,morris16,tortora16}, and hidden truncation hyperbolic distributions \citep{murray17}. 
A comprehensive review of model-based clustering work, up to and including some recent work on non-Gaussian mixtures, is given by \cite{mcnicholas16b}.

Unobserved or missing observations are frequently a hindrance in multivariate datasets and so developing mixture models that can accommodate incomplete data is an important issue in model-based clustering. The maximum likelihood and Bayesian approaches are two common imputation paradigms for analyzing data with incomplete observations. Herein, the missing data mechanism is assumed to be missing at random (MAR), as per \citet{rubin76} and \citet{little87}, meaning that the probability that a variable is missing for a particular individual depends only on the observed data and not on the value of the missing variable. Note that missing completely at random (MCAR) is a special case of MAR. Under MAR, the missing data mechanisms are ignorable for methods using the maximum likelihood approach. 

The maximum likelihood approach to clustering incomplete data has been well studied and is often used, particularly for Gaussian mixture models \citep[e.g.,][]{ghahramani94,lin06,browne13}. \citet{wang04} present a framework maximum likelihood estimation using an expectation-maximization (EM) algorithm \citep{dempster77} to fit a mixture of multivariate $t$-distributions with arbitrary missing data patterns, which was generalized by \cite{lin09b} to efficient supervised learning via the parameter expanded (PX-EM) algorithm \citep{liu98} through two auxiliary indicator matrices. \citet{linlearning14} further develops a family of multivariate-$t$ mixture models with 14 eigen-decomposed scale matrices in the presence of missing data through a computationally flexible EM algorithm by incorporating two auxiliary indicator matrices. \cite{wang15} uses a formulation of the mixture of skew-t distributions for model-based clustering with missing data. 

We consider fitting mixtures of generalized hyperbolic distributions (MGHD) and mixtures of multivariate skew-t distributions (MST) with missing information. In each case, an EM algorithm is used for model selection. The chosen formulation of the (multivariate) generalized hyperbolic distribution (GHD) is that used by \cite{browne15} and has formulations of several well-known distributions as special cases such as the multivariate skew-$t$, normal inverse Gaussian, variance-gamma, Laplace, and Gaussian distributions \citep[cf.][]{mcneil05}. In addition to considering missing data, we develop families of MGHD and MST mixture models, each with 14 parsimonious eigen-decomposed scale matrices corresponding to the famous Gaussian parsimonious clustering models of \citep[GPCMs;][]{banfield93,celeux95}; see Table~\ref{tab:GPCM} (Appendix~\ref{sec:gpcm}).


\section{Background}
\label{sec:review}

\subsection{Generalized Inverse Gaussian Distribution}
A random variable $W \in \R^{+}$ is said to have a generalized inverse Gaussian (GIG) distribution, introduced by \citep{good53}, with parameters $\lambda$, $\chi$, and $\psi$ if its probability density function is given by
\begin{equation}
\label{eqn:gig}
f_{\text{GIG}}(w~|~\lambda,\chi,\psi) = \frac{(\psi/\chi)^{\lambda/2}w^{\lambda-1}}{2K_\lambda(\sqrt{\psi\chi})}\text{exp}\left\{-\frac{\psi w+\chi/w}{2}\right\},
\end{equation}
where $\psi, \chi \in \R^{+}, \lambda \in \R$, and $K_\lambda$ is the modified Bessel function of the third kind with index $\lambda$. Herein, we write $W \sim \text{GIG}(\lambda,\chi,\psi)$ to indicate that a random variable $W$ has the GIG density as parameterized in \eqref{eqn:gig}. The GIG distribution has some attractive properties \citep{barndorff77b,blaesild78,halgreen79,jorgensen82}, including the tractability of the expectations:
\begin{equation}
\label{eqn:expgig}
\E [W^\alpha] = \left(\frac{\chi}{\psi}\right)^{\alpha/2}\frac{K_{\lambda+\alpha}(\sqrt{\psi\chi})}{K_{\lambda}(\sqrt{\psi\chi})},
\end{equation}
for $\alpha \in \R$, and
\begin{equation}\label{eqn:expgiglog}
\E [\log W] = \log\left(\sqrt{\frac{\chi}{\psi}}\right)+\frac{\partial}{\partial\lambda}\log(K_\lambda(\sqrt{\psi\chi})).
\end{equation}
Specifically, for $\alpha=1$ and $\alpha=-1$, we have
\begin{align}
\nonumber
\E [W] &= \sqrt{\frac{\chi}{\psi}}\frac{K_{\lambda+1}(\sqrt{\psi\chi})}{K_{\lambda}(\sqrt{\psi\chi})},\\
\nonumber
\E [{1}/{W}] &= \sqrt{\frac{\psi}{\chi}}\frac{K_{\lambda-1}(\sqrt{\psi\chi})}{K_{\lambda}(\sqrt{\psi\chi})} = \sqrt{\frac{\psi}{\chi}}\frac{K_{\lambda+1}(\sqrt{\psi\chi})}{K_{\lambda}(\sqrt{\psi\chi})}-\frac{2\lambda}{\chi}.
\nonumber
\end{align}

\citet{browne15} introduce another parameterization of the GIG distribution by setting $\omega=\sqrt{\psi\chi}$ and $\eta=\sqrt{\chi/\psi}$. Write $W \sim \mathcal{I}(\lambda,\eta,\omega)$; its density is given by
\begin{equation}
f_{\mathcal{I}} (w \mid \lambda,\eta,\omega) = \frac{(w/\eta)^{\lambda-1}}{2\eta K_{\lambda}(\omega)} \text{exp} \left\{ -\frac{\omega}{2}\left(\frac{w}{\eta}+\frac{\eta}{w}\right) \right\}
\end{equation}
for $w>0$, where $\eta \in \R^{+}$ is a scale parameter and $\omega \in \R^{+}$ is a concentration parameter. These two parameterizations of the GIG distribution are important ingredients for building the generalized hyperbolic distribution presented later.

\subsection{Generalized Hyperbolic Distribution }
Several alternative parameterizations of the GHD have appeared in the literature, e.g., \cite{blaesild81}, \cite{mcneil05}, and \cite{browne15}. \citet{barndorff77}  introduces the generalized hyperbolic distribution (GHD) to model the distribution of the sand grain sizes and subsequent reports described its statistical properties~\citep[e.g.,][]{barndorff78, blaesild81}. However, under this standard parameterization, the parameters of the mixing distribution are not invariant by affine transformations.  An important innovation was made by~\citet{mcneil05}, who gave a new parameterization of the GHD. Under this new parameterization, the linear transformation of GHD remains in the same sub-family characterized by the parameters of the mixing distribution. However, there is an identifiability issue arising under this parameterization. To solve this problem,~\citet{browne15} give an alternative parameterization.

Following~\citet{mcneil05}, a $p\times 1$ random vector $\vecY$ is said to follow a generalized hyperbolic distribution with index parameter $\lambda$, concentration parameters $\chi$ and $\psi$, location vector $\vecmu$, dispersion matrix $\matsig$, and skewness vector $\vecalpha$, denoted by $\vecY \sim \text{GH}_p(\lambda,\chi,\psi,\vecmu,\matsig,\vecalpha)$, if it can be represented by
\begin{equation}
\vecY = \vecmu + W\vecalpha + \sqrt{W}\vecU,
\end{equation}
where $\vecU \bot W$, $W \sim \text{GIG}(\lambda,\chi,\psi)$, $\vecU \sim \mathcal{N}(\mathbf{0},\matsig)$, and the symbol $\bot$ indicates independence. It follows that $\vecY \mid w \sim \mathcal{N} (\vecmu + w\vecalpha, w\matsig)$. So, the density of the generalized hyperbolic random vector $\vecY$ is given by
\begin{equation}
	\label{eqn:ghd}
	f(\vecy \mid \varthet) = \left [\frac{\chi +\delta(\vecy,\vecmu \mid \matsig )}{\psi+\vecalpha^{\intercal}\matsig\inv\vecalpha}\right]^{\frac{\lambda-p/2}{2}}\frac{(\psi/\chi)^{\lambda/2}K_{\lambda-p/2}\left(\sqrt{(\chi +\delta(\vecy,\vecmu \mid \matsig))(\psi+\vecalpha^{\intercal}\matsig\inv\vecalpha)}\right)}{(2\pi)^{p/2} |\matsig|^{1/2}K_\lambda(\sqrt{\chi\psi})\exp\{-(\vecy-\vecmu)^{\intercal}\matsig\inv\vecalpha\}},
\end{equation}
where $\delta (\vecy, \vecmu \mid \matsig) = (\vecy - \vecmu)^{\intercal}\matsig\inv(\vecy - \vecmu)$ is the squared Mahalanobis distance between $\vecy$ and $\vecmu$, $K_\lambda$ is the modified Bessel function of the third kind with index $\lambda$, and $\varthet= (\lambda, \chi,\psi, \vecmu, \matsig,\vecalpha)$ denotes the model parameters. The mean and covariance matrix of $\vecY$ are
\begin{equation}\label{eqn:meanvary}
\E (\vecY) = \vecmu+\E (W)\vecalpha \quad \text{and} \quad \mathbb{V}\text{ar}(\vecY) = \E (W)\matsig + \mathbb{V}\text{ar}(W) \vecalpha \vecalpha^{\intercal},
\end{equation}
respectively, where $\E (W)$ and $\mathbb{V}\text{ar}(W)$ are the mean and variance of the random variable~$W$, respectively.

Note that, in this parameterization, we need to hold $|\matsig|=1$ to ensure identifiability. Using $|\matsig|=1$ solves the identifiability problem but would be prohibitively restrictive for model-based clustering and classification applications. Hence,~\citet{browne15} develop a new parameterization of the GHD with index parameter $\lambda$, concentration parameter $\omega$, location vector $\vecmu$, dispersion matrix $\matsig$, and skewness vector $\vecbeta=\eta\vecalpha$, denoted by $\vecY \sim \text{GHD}_p(\lambda,\omega,\vecmu,\matsig,\vecbeta)$. Note that $\eta=1$. This formulation is given by 
\begin{equation}
\label{eqn:ghd2}
\vecY = \vecmu + W\vecbeta + \sqrt{W}\vecU,
\end{equation}
where $\vecU \bot W$, $W \sim \text{GIG}(\omega/\eta,\omega\eta,\lambda)$, with $\eta=1$, and $\vecU \sim \mathcal{N}(\mathbf{0},\matsig)$. Under this parameterization, the density of the generalized hyperbolic random vector $\vecY$ is
\begin{equation}
\label{eqn:ghd3}
	f(\vecy \mid \varthet) = \left [\frac{\omega +\delta(\vecy,\vecmu \mid \matsig )}{\omega+\vecbeta^{\intercal}\matsig\inv\vecbeta}\right]^{\frac{\lambda-p/2}{2}}\frac{K_{\lambda-p/2}\left(\sqrt{(\omega +\delta(\vecy,\vecmu \mid \matsig))(\omega+\vecbeta^{\intercal}\matsig\inv\vecbeta)}\right)}{(2\pi)^{p/2} |\matsig|^{1/2}K_\lambda(\omega)\text{exp}\{-(\vecy-\vecmu)^{\intercal}\matsig\inv\vecbeta\}},
\end{equation}
where $\delta(\vecy,\vecmu \mid \matsig )$ and $K_{\lambda-p/2}$ are as described earlier. We use this parameterization when we describe parameter estimation (cf.\ Section~\ref{sec:mghd}).

The following result shows an appealing closure property of the generalized hyperbolic distribution under affine transformation and conditioning as well as the formation of marginal distributions, which is useful for developing new methods presented later. Suppose that $\vecY$ is a $p$-dimensional random vector having a generalized hyperbolic distribution as in \eqref{eqn:ghd3}, i.e., $\vecY \sim \text{GHD}_p(\lambda,\omega,\vecmu,\matsig,\vecbeta)$. Assume that $\vecY$ is partitioned as $\vecY = (\vecY_1^{\intercal}, \vecY_2^{\intercal})^{\intercal}$, where $\vecY_1$ takes values in $\R^{d_1}$ and $\vecY_2$ in $\R^{d_1}=\R^{p-d_1}$, with
\begin{align}
\nonumber
\vecmu = \begin{pmatrix} \vecmu_1 \\ \vecmu_2 \end{pmatrix},&&\vecbeta = \begin{pmatrix} \vecbeta_1 \\ \vecbeta_2 \end{pmatrix},&&\matsig = \begin{pmatrix} \matsig_{11} & \matsig_{12}\\ \matsig_{21} &\matsig_{22} \end{pmatrix},
\end{align}
where $\vecY$, $\vecmu$, and $\vecbeta$ have similar partitions. Furthermore, $\matsig_{11}$ is $d_1\times d_1$ and $\matsig_{22}$ is $d_2\times d_2$.

\begin{prop}
Affine transformation of the generalized hyperbolic distribution. If $\vecY \sim \text{GHD}_p(\lambda,\omega,\vecmu,\matsig,\vecbeta)$ and $\vecX = \vecB\vecY+\vecb$ where $\vecB \in \R^{k\times p}$ and $\vecb \in \R^{p}$, then
\begin{equation}
\vecX \sim \text{GHD}_k(\lambda,\omega,\vecB\vecmu+\vecb,\vecB\matsig\vecB^{\intercal},\vecB\vecbeta),
\end{equation}
\end{prop}\begin{proof}
The result follows by substituting  \eqref{eqn:ghd2} into $\vecX = \vecB\vecY+\vecb$.
\end{proof}

\begin{prop} The marginal distribution of $\vecY_1$ is a generalized hyperbolic distribution as in \eqref{eqn:ghd3} with index parameter $\lambda$, concentration parameter $\omega$, location vector $\vecmu_1$, dispersion matrix $\matsig_{11}$, and skewness vector $\vecbeta_1$, i.e., $\vecY_1 \sim \text{GHD}_{d_1}(\lambda,\omega,\vecmu_1,\matsig_{11},\vecbeta_1)$.
\end{prop}\begin{proof}
The result follows by applying Proposition 1 and choosing $\vecB = [ \mathbf{I}_{d_1},\mathbf{0}$] and $\vecb=\mathbf{0}$. The parameters $\lambda,\omega$ inherited from the mixing distribution $W \sim \mathcal{I}(\lambda,\eta=1,\omega)$ remain the same under the affine transformation and marginal distribution.
\end{proof}

\begin{prop} The conditional distribution of $\vecY_2$ given $\vecY_1 = \vecy_1$ is a generalized hyperbolic distribution as in \eqref{eqn:ghd}, i.e., $\vecY_2 \mid \vecY_1=\vecy_1 \sim \text{GH}_{d_2}(\lambda_{
2\mid1},\chi_{2\mid1},\psi_{2\mid1},\vecmu_{2\mid1},\matsig_{2\mid1},\vecbeta_{2\mid1})$, where 
\begin{align}
\nonumber
\lambda_{2\mid1}&= \lambda-\frac{d_1}{2},&\chi_{2\mid1}&=\omega+(\vecy_1-\vecmu_1)^{\intercal}\matsig_{11}\inv(\vecy_1-\vecmu_1),\\
\nonumber
\psi_{2\mid1}&=\omega+\vecbeta_1^{\intercal}\matsig_{11}^{\intercal}\vecbeta,&\vecmu_{2\mid1}&=\vecmu_2+\matsig_{12}^{\intercal}\matsig_{11}\inv(\vecy_1-\vecmu_1),\\
\nonumber
\matsig_{2\mid1}&=\matsig_{22}-\matsig_{12}^{\intercal}\matsig_{11}\inv\matsig_{12},&\vecbeta_{2\mid1}&=\vecbeta_2-\matsig_{12}^{\intercal}\matsig_{11}\inv\vecbeta_1.
\end{align}\end{prop}
The proof of Proposition 3 is given in Appendix~B.

\subsection{The Multivariate Skew-$t$ Distribution}
There are several alternative formulations of multivariate skew-t distributions appearing in the literature \citep*[e.g.,][]{branco01, sahu03, murray14a, lee14}. \citet{lin11} develop a mixture of multivariate skew-t distributions incomplete data using the formulation of \citet{sahu03}. Herein, the formulation of the multivariate skew-$t$ distribution arising from the generalized hyperbolic distribution is used. This formulation of the multivariate skew-$t$ distribution has been used by \citet{murray14a} to develop a mixture of skew-$t$ factor analyzers model.

Following \citet{mcneil05}, a $p$ x $1$ random vector $\vecY$ is said to follow a multivariate skew-t distribution with degree of freedom parameter $v$, location vector $\vecmu$, dispersion matrix $\matsig$, and skewness vector $\vecbeta$, denoted by $\vecY \sim \text{ST}_p(v,\vecmu,\matsig,\vecbeta)$, if it can be represented by
\begin{equation}
\label{eqn:skewtstochastic}
\vecY = \vecmu + W\vecbeta + \sqrt{W}\vecU,
\end{equation}
where $\vecU \bot W$, $W \sim \text{IG}(v/2,v/2)$, $\vecU \sim \mathcal{N}(\mathbf{0},\matsig)$, with $\text{IG}(\cdot)$ denoting the inverse Gamma distribution. It follows that $\vecY \mid w \sim \mathcal{N} (\vecmu + w\vecbeta, w\matsig)$ and the pdf of the multivariate skew-t random vector $\vecY$ is given by
\begin{equation}
	\label{eqn:skewt}
	f(\vecy \mid \varthet) = \left [\frac{v+\delta(\vecy,\vecmu \mid \matsig )}{\vecbeta^{\intercal}\matsig\inv\vecbeta}\right]^{\frac{-v-p}{4}}\frac{v^{v/2}K_{(-v-p)/2}\left(\sqrt{(v +\delta(\vecy,\vecmu \mid \matsig))(\vecbeta^{\intercal}\matsig\inv\vecbeta)}\right)}{(2\pi)^{p/2} |\matsig|^{1/2}\Gamma(v/2)2^{v/2-1}\text{exp}\{-(\vecy-\vecmu)^{\intercal}\matsig\inv\vecbeta\}}.
\end{equation}
This formulation of the multivariate skew-t distribution can be obtained as a special case of the generalized hyperbolic distribution by setting $\lambda=-v/2$ and $\chi=v$, and letting $\psi \to 0$.
Similarly, this formulation of the multivariate skew-t distribution has a closed form under affine transformation and conditioning, and the formation of marginal distributions, which is useful for developing new methods presented later. Suppose that $\vecY$ is a $p$-dimensional random vector having the multivariate skew-t distribution as in \eqref{eqn:skewt}, i.e., $\vecY \sim \text{ST}_p(v,\vecmu,\matsig,\vecbeta)$. Assume that $\vecY$ is partitioned as $\vecY = (\vecY_1^{\intercal}, \vecY_2^{\intercal})^{\intercal}$, where $\vecY_1$ takes values in $\R^{d_1}$ and $\vecY_2$ in $\R^{d_1}=\R^{p-d_1}$, with
\begin{align}
\nonumber
\vecmu = \begin{pmatrix} \vecmu_1 \\ \vecmu_2 \end{pmatrix}&&\vecbeta = \begin{pmatrix} \vecbeta_1 \\ \vecbeta_2 \end{pmatrix}&&\matsig = \begin{pmatrix} \matsig_{11} & \matsig_{12}\\ \matsig_{21} &\matsig_{22} \end{pmatrix},
\end{align}
where $\vecY$, $\vecmu$, and $\vecbeta$ have similar partitions. Furthermore, $\matsig_{11}$ is $d_1\times d_1$ and $\matsig_{22}$ is $d_2\times d_2$.

\begin{prop} Affine transformation of the multivariate skew-t distribution. If $\vecY \sim \text{ST}_p(v,\vecmu,\matsig,\vecbeta)$ and $\vecX = \vecB\vecY+\vecb$, where $\vecB \in \R^{k\times p}$ and $\vecb \in \R^{p}$, then
\begin{equation}
\vecX \sim \text{ST}_k(v,\vecB\vecmu+\vecb,\vecB\matsig\vecB^{\intercal},\vecB\vecbeta).
\end{equation}
\end{prop}\begin{proof}
The proof follows easily by substituting  \eqref{eqn:skewtstochastic} into $\vecX = \vecB\vecY+\vecb$.
\end{proof}

\begin{prop} The marginal distribution of $\vecY_1$ is a multivariate skew-t distribution as in \eqref{eqn:skewt} with degree of freedom parameter $v$, location vector $\vecmu_1$, dispersion matrix $\matsig_{11}$, and skewness vector $\vecbeta_1$, i.e., $\vecY_1 \sim \text{ST}_{d_1}(v,\vecmu_1,\matsig_{11},\vecbeta_1)$.
\end{prop}\begin{proof}
The proof follows easily by applying Proposition 4 and choosing $\vecB = [ \mathbf{I}_{d_1},\mathbf{0}$] and $\vecb=\mathbf{0}$. The degree of freedom parameter $v$ inherited from the mixing distribution $W \sim\text{IG}(v/2,v/2)$ remains invariant under affine transformation and marginal distribution.
\end{proof}

\begin{prop} The conditional distribution of $\vecY_2$ given $\vecY_1 = \vecy_1$ is a generalized hyperbolic distribution as in \eqref{eqn:ghd}, i.e., $\vecY_2 \mid \vecy_1 \sim \text{GH}_{d_2}(\lambda_{
2\mid1},\chi_{2\mid1},\psi_{2\mid1},\vecmu_{2\mid1},\matsig_{2\mid1},\vecbeta_{2\mid1})$, where 
\begin{align}
\nonumber
\lambda_{2\mid1}&= -(v+d_1)/2,&\chi_{2\mid1}&=v+(\vecy_1-\vecmu_1)^{\intercal}\matsig_{11}\inv(\vecy_1-\vecmu_1),\\
\nonumber
\psi_{2\mid1}&=\vecbeta_1^{\intercal}\matsig_{11}^{\intercal}\vecbeta,&\vecmu_{2\mid1}&=\vecmu_2+\matsig_{12}^{\intercal}\matsig_{11}\inv(\vecy_1-\vecmu_1),\\
\nonumber
\matsig_{2\mid1}&=\matsig_{22}-\matsig_{12}^{\intercal}\matsig_{11}\inv\matsig_{12},&\vecbeta_{2\mid1}&=\vecbeta_2-\matsig_{12}^{\intercal}\matsig_{11}\inv\vecbeta_1.
\end{align}\end{prop}
The proof of Proposition 6 is similar to that for Proposition 3, hence is omitted. Similar results for Proposition~4, 5, and 6 have been obtained in \citet{arellano10}. 

\section{MGHD with Incomplete Data}\label{sec:mghd}
Let $\vecY_1,\ldots,\vecY_n$ be $p$-dimensional random variables arising from a heterogeneous population with $G$ disjoint MGHD subpopulations. That is, each $\vecY_i$ has the density
\begin{equation}
\label{eqn:mghd}
f_{\text{MGHD}}(\vecy_i\mid\matTheta) = \gsum \pi_g f_{\text{GHD}} (\vecy_i\mid\lambda_g,\omega_g,\vecmu_g,\matsig_g,\vecbeta_g),
\end{equation}
where $\pi_g>0$, such that $\gsum \pi_g=1$, are the mixing proportions, $\matTheta$ denotes the model parameters, 
and $f_{\text{GHD}} (\vecY_i\mid\lambda_g,\omega_g,\vecmu_g,\matsig_g,\vecbeta_g)$ is the GHD density defined in \eqref{eqn:ghd3}.

To apply the MGHD model \eqref{eqn:mghd} in the clustering paradigm, introduce $\vecz_{i} = (z_{i1},\ldots,z_{ig})^{\intercal}$, where $z_{ig}=1$ if observation $i$ is in component $g$ and $z_{ig}=0$ otherwise. The corresponding random variable $\vecZ_i\sim\mathcal{M}(1;\pi_1,\ldots,\pi_G)$, i.e., $\vecZ_i$ follows a multinomial distribution with one trial and cell probabilities $\pi_1,\ldots,\pi_G$.

A three-level hierarchical representation of the MGHD model \eqref{eqn:mghd} can be expressed by
\begin{align}
\nonumber
\vecY_i\mid w_{ig},z_{ig}=1 &\sim \mathcal{N}(\vecmu_g+w_{ig}\vecbeta_g,w_{ig}\matsig_g),\\
\label{eqn:mghdthreelevel}
W_{ig}\mid z_{ig}=1&\sim \mathcal{I}(\lambda_g,\eta=1,\omega_g),\\
\nonumber
\vecZ_i&\sim\mathcal{M}(1;\pi_1,\ldots,\pi_G).
\end{align}
The complete-data consist of the observed $\vecy_i$ together with the missing group membership $z_{ig}$ and the latent $w_{ig}$, for $i=1,\ldots,n$ and $g=1,\ldots,G$, and the complete-data log-likelihood is given by
\begin{equation}
\label{eqn:mghdloglik}
l_{\text{c}}(\matTheta) = \isum\gsum z_{ig}\left[\log\pi_g+\log\phi(\vecy_i\mid\vecmu_g+w_{ig}\vecbeta_g,w_{ig}\matsig_g)+\log h(w_{ig}\mid\lambda_g,\omega_g)\right].
\end{equation}

\citet{browne15} present an EM algorithm for parameter estimation with the MGHD when there is no missing data in $\vecy_1,\ldots,\vecy_n$. We are interested in parameter estimation for the MGHD model \eqref{eqn:mghd} when $\vecy_1,\ldots,\vecy_n$ are partially observed with arbitrary missing patterns. The missing data mechanism is assumed to be MAR. Assume now that we split $\vecy_i$ into two components, $\vecy_i^{\text{o}}$ and $\vecy_i^{\text{m}}$ that denote the observed and missing components of $\vecy_i$, respectively. In general, each data vector $\vecy_i$ may have a different pattern of missing features, i.e., $\vecy_i=(\vecy_i^{\text{o}_i\intercal},\vecy_i^{\text{m}_i\intercal})^{\intercal}$, but can be simplified for the sake of clarity. 

For each $\vecy_i=(\vecy_i^{\text{o}\intercal},\vecy_i^{\text{m}\intercal})^{\intercal}$, partition the vector mean $\vecmu_g=(\vecmu_{g,i}^{\text{o}\intercal},\vecmu_{g,i}^{\text{m}\intercal})^{\intercal}$, where $\vecmu_{g,i}^{\text{o}}$ and $\vecmu_{g,i}^{\text{m}}$ denote the sub-vectors of $\vecmu_g$ matching the observed and missing components of $\vecy_i$, respectively. Similarly, the skewness vector is $\vecbeta_g=(\vecbeta_{g,i}^{\text{o}\intercal},\vecbeta_{g,i}^{\text{m}\intercal})^{\intercal}$ and  the covariance matrix $\matsig_g$ as 
\begin{equation}
\label{eqn:matsigmiss}
\matsig_g=\begin{pmatrix}\matsig_{g,i}^{\text{oo}}&\matsig_{g,i}^{\text{om}}\\\matsig_{g,i}^{\text{mo}}&\matsig_{g,i}^{mm}\end{pmatrix} \text{and}\quad 
\matsig_g\inv=\begin{pmatrix}(\matsig_{g,i}^{\text{oo}})^{-1}&(\matsig_{g,i}^{\text{om}})^{-1}\\
(\matsig_{g,i}^{\text{mo}})^{-1}&(\matsig_{g,i}^{\text{mm}})^{-1}\end{pmatrix},
\end{equation}
correspond to $\vecy_i=(\vecy_i^{\text{o}\intercal},\vecy_i^{\text{m}\intercal})^{\intercal}$. As a result, in addition to the observed $\vecy_i^{\text{o}}$, the missing group membership $z_{ig}$, and the latent variable $w_{ig}$, the complete-data also include the missing data $\vecy_i^{\text{m}}$. In the framework of the EM algorithm, the missing data $\vecy_i^{\text{m}}$ are considered to be random variables that are updated in each iteration. Hence, the complete-data log-likelihood \eqref{eqn:mghdloglik} is rewritten as
\begin{align}
\label{eqn:mghdloglikmiss}
\nonumber
l_{\text{c}}(\matTheta) = \isum\gsum z_{ig}\big[\log\pi_g+&\log\phi(\vecy_i^{\text{o}},\vecy_i^{\text{m}}\mid\vecmu_g+w_{ig}\vecbeta_g,w_{ig}\matsig_g)
+\log h_{\mathcal{I}}(w_{ig}\mid\lambda_g,\omega_g)\big].\end{align}
Given (\ref{eqn:mghdthreelevel}), we establish the following:
\begin{itemize}
\item The marginal distribution of $\vecY_i^{\text{o}}$ given is 
	\begin{equation}
	\nonumber
	\vecY_i^{\text{o}} \sim \gsum \pi_g f_{\text{GHD},p_i^{\text{o}}}(\lambda_g,\omega_g,\vecmu_{g,i}^{\text{o}}, \matsig_{g,i}^{\text{oo}}, \vecbeta_{g,i}^{\text{o}}),
	\end{equation}
	where $p_i^{\text{o}}$ is the dimension corresponding to the observed component $\vecy_i^{\text{o}}$, which should be exactly written as $p_i^{\text{o}_i}$ but here is simplified.
\item The conditional distribution of $\vecY_i^{\text{m}}$ given $\vecy_i^{\text{o}}$ and $z_{ig}=1$, according to Proposition~3, is 
	\begin{equation}
	\label{eqn:ymissconpdf}
	\vecY_i^{\text{m}} \mid \vecy_i^{\text{o}}, z_{ig}=1\sim \text{GH}_{p-p_i^{\text{o}}}\left(\lambda_{g,i}^{\text{m}\mid \text{o}},\chi_{g,i}^{\text{m}\mid \text{o}},\psi_{g,i}^{\text{m}\mid \text{o}},\vecmu_{g,i}^{\text{m}\mid \text{o}},\matsig_{g,i}^{\text{m}\mid \text{o}},\vecbeta_{g,i}^{\text{m}\mid \text{o}}\right), 
	\end{equation}
	where
	\begin{align*}
	\lambda_{g,i}^{\text{m}\mid \text{o}}&= \lambda_g-\frac{p_i^{\text{o}}}{2},&\chi_{g,i}^{\text{m}\mid \text{o}}&=\omega_g+(\vecy_i^{\text{o}}-\vecmu_{g,i}^{\text{o}})^{\intercal}(\matsig_{g,i}^{\text{oo}})\inv(\vecy_i^{\text{o}}-\vecmu_{g,i}^{\text{o}}),\\
	\psi_{g}^{\text{m}\mid \text{o}}&=\omega_g+(\vecbeta_{g,i}^{\text{o}})^{\intercal}(\matsig_{g,i}^{\text{oo}})\inv\vecbeta_{g,i}^{\text{o}},&\vecmu_{g,i}^{\text{m}\mid \text{o}}&=\vecmu_g^{\text{m}}+(\matsig_{g,i}^{\text{om}})^{\intercal}(\matsig_{g,i}^{\text{oo}})\inv(\vecy_i^{\text{o}}-\vecmu_{g,i}^{\text{o}}),\\
	\matsig_{g,i}^{\text{m}\mid \text{o}}&=\matsig_{g,i}^{\text{mm}}-(\matsig_{g,i}^{\text{om}})^{\intercal}(\matsig_{g,i}^{\text{oo}})\inv\matsig_{g,i}^{\text{om}},&\vecbeta_{g,i}^{\text{m}\mid \text{o}}&=\vecbeta_{g,i}^{\text{m}}-(\matsig_{g,i}^{\text{om}})^{\intercal}(\matsig_{g,i}^{\text{oo}})\inv\vecbeta_{g,i}^{\text{o}}.
	\end{align*}
\item The conditional distribution of $\vecY_i^{\text{m}}$ given $\vecy_i^{\text{o}}, w_{ig}$, and $z_{ig}=1$ is 
        \begin{equation}
        \vecY_i^{\text{m}} \mid \vecy_i^{\text{o}}, w_{ig}, z_{ig}=1\sim \mathcal{N}_{p-p_i^{\text{o}}}(\vecmu_{g,i}^{\text{m}\mid \text{o}}+w_{ig}\vecbeta_{g,i}^{\text{m}\mid \text{o}}, w_{ig}\matsig_{g,i}^{\text{m}\mid \text{o}}).
        \end{equation}
\item The conditional distribution of $W_i$ given $\vecy_i^{\text{o}}$ and $z_{ig}=1$ is
	\begin{equation}
	\label{eqn:wconpdf}
	W_{ig} \mid \vecy_i^{\text{o}},z_{ig}=1 \sim \text{GIG}\left(\omega_g+(\vecbeta_{g,i}^{\text{o}})^{\intercal}(\matsig_{g,i}^{\text{oo}})\inv\vecbeta_{g,i}^{\text{o}}, \omega_g+\delta(\vecy_i^{\text{o}},\vecmu_{g,i}^{\text{o}}\mid \matsig_{g,i}^{\text{oo}}), \lambda_g-\frac{p_i^{\text{o}}}{2}\right).
	\end{equation} 
\end{itemize}
After a little algebra, we get the complete data log-likelihood function is
\begin{equation}
\label{eqn:loglikelihood}
\begin{split}
&l_{\text{c}}(\matTheta) =\isum\gsum z_{ig}\log\pi_g+\isum\gsum z_{ig}\left[ -\frac{p}{2}\log(2\pi)-\frac{p}{2}\log w_{ig}+\frac{1}{2}\log |\matsig_g\inv|\right]\\
&-\frac{1}{2}\isum\gsum\text{tr}\left( \matsig_g\inv z_{ig} \frac{1}{w_{ig}}\begin{pmatrix} (\vecy_i^{\text{o}}-\vecmu_{g,i}^{\text{o}})(\vecy_i^{\text{o}}-\vecmu_{g,i}^{\text{o}})^{\intercal} &(\vecy_i^{\text{o}}-\vecmu_{g,i}^{\text{o}})(\vecy_i^{\text{m}}-\vecmu_{g,i}^{\text{m}})^{\intercal}\\(\vecy_i^{\text{m}}-\vecmu_{g,i}^{\text{m}})^{\intercal}(\vecy_i^{\text{o}}-\vecmu_{g,i}^{\text{o}})&(\vecy_i^{\text{m}}-\vecmu_{g,i}^{\text{m}})(\vecy_i^{\text{m}}-\vecmu_{g,i}^{\text{m}})^{\intercal}\end{pmatrix} \right)\\
&+\frac{1}{2}\isum\gsum\text{tr}\left(  \matsig_g\inv z_{ig} \begin{pmatrix}\vecbeta_{g,i}^{\text{o}}\\\vecbeta_{g,i}^{\text{m}}\end{pmatrix}\begin{pmatrix} (\vecy_i^{\text{o}}-\vecmu_{g,i}^{\text{o}})^{\intercal}&(\vecy_i^{\text{m}}-\vecmu_{g,i}^{\text{m}})^{\intercal}\end{pmatrix}\right)\\
&+\frac{1}{2}\isum\gsum\text{tr}\left(  \matsig_g\inv z_{ig} \begin{pmatrix} \vecy_i^{\text{o}}-\vecmu_{g,i}^{\text{o}}\\\vecy_i^{\text{m}}-\vecmu_{g,i}^{\text{m}}\end{pmatrix}\begin{pmatrix}(\vecbeta_{g,i}^{\text{o}})^{\intercal}&(\vecbeta_{g,i}^{\text{m}})^{\intercal}\end{pmatrix}\right)-\frac{1}{2}\isum\gsum z_{ig}w_{ig}\vecbeta_{g,i}^{\intercal}\matsig_g\inv\vecbeta_{g,i}\\
&+\isum\gsum z_{ig}\left[(\lambda_g-1)\log w_{ig}-\log(2K_{\lambda_g}(\omega_g))-\frac{\omega_g}{2}\left(w_{ig}+\frac{1}{w_{ig}}\right)\right].
\end{split}
\end{equation}

On the $k$th iteration of the E-step, the expected value of the complete data log-likelihood is computed given the observed data $\vecy_1^{\text{o}},\ldots,\vecy_n^{\text{o}}$ and the current parameter updates $\matTheta^{(k)}$. That is, we need to compute $\E(Z_{ig}\mid \vecy_i^{\text{o}};\matTheta^{(k)})$, $\E(W_{ig}\mid \vecy_i^{\text{o}},z_{ig}=1;\matTheta^{(k)})$, $\E(\log W_{ig}\mid\vecy_i^{\text{o}},z_{ig}=1;\matTheta^{(k)})$, $\E({1}/{W_{ig}}\mid \vecy_i^{\text{o}},z_{ig}=1;\matTheta^{(k)})$, $\E(\vecY_i^{\text{m}}\mid\vecy_i^{\text{o}},z_{ig}=1,w_i;\matTheta^{(k)})$, and $\E(\vecY_i^{\text{m}} (\vecY_i^{\text{m}})^{\intercal}\mid\vecy_i^{\text{o}},z_{ig}=1,w_i;\matTheta^{(k)})$.

First, let $\hat{z}_{ig}^{(k)}$ denote the \textit{a~posteriori} probability that $i$-th observation belongs to the $g$-th component of the mixture, based on the observed data:
\begin{equation}
\nonumber
\hat{z}_{ig}^{(k)} \colonequals \E(Z_{ig}\mid \vecy_i^{\text{o}},\matTheta^{(k)})=\frac{\pi_g^{(k)}f_{\text{GHD},p_i^{\text{o}}}(\vecy_i^{\text{o}};\lambda_g^{(k)},\omega_g^{(k)},\vecmu_{g,i}^{\text{o}(k)},\matsig_{g,i}^{\text{oo}(k)},\vecbeta_{g,i}^{\text{o}(k)})}{\sum_{l=1}^{G}\pi_l^{(k)}f_{\text{GHD},p_i^{\text{o}}}(\vecy_i^{\text{o}};\lambda_l^{(k)},\omega_l^{(k)},\vecmu_{l,i}^{\text{o}(k)},\matsig_{l,i}^{\text{oo}(k)},\vecbeta_{l,i}^{\text{o}(k)})}.
\end{equation}
Given (\ref{eqn:expgig}), (\ref{eqn:expgiglog}), and (\ref{eqn:wconpdf}), we have the following expectations as to the latent variable $W$:
\begin{align*}
a_{ig}^{(k)}&\colonequals\E(W_{ig}\mid \vecy_i^{\text{o}},z_{ig}=1;\matTheta^{(k)}) = \sqrt{\frac{\omega_g^{(k)}+\delta(\vecy_i^{\text{o}},\vecmu_{g,i}^{\text{o}(k)}\mid\matsig_{g,i}^{\text{oo}(k)})}{\omega_g^{(k)}+\vecbeta_{g,i}^{o(k)\intercal}(\matsig_{g,i}^{\text{oo}(k)})\inv\vecbeta_{g,i}^{\text{o}(k)}}}\\
&\qquad\times\frac{K_{\lambda_g^{(k)}-\frac{p_i^0}{2}+1}\left(\sqrt{(\omega_g^{(k)}+\delta(\vecy_i^{\text{o}},\vecmu_{g,i}^{\text{o}(k)}\mid\matsig_{g,i}^{\text{oo}(k)}))(\omega_g^{(k)}+(\vecbeta_{g,i}^{\text{o}(k)})^{\intercal}(\matsig_{g,i}^{\text{oo}(k)})\inv\vecbeta_{g,i}^{\text{o}(k)})}\right)}{K_{\lambda_g^{(k)}-\frac{p_i^0}{2}}\left(\sqrt{(\omega_g^{(k)}+\delta(\vecy_i^{\text{o}},\vecmu_{g,i}^{\text{o}(k)}\mid\matsig_{g,i}^{\text{oo}(k)}))(\omega_g^{(k)}+(\vecbeta_{g,i}^{\text{o}(k)})^{\intercal}(\matsig_{g,i}^{\text{oo}(k)})\inv\vecbeta_{g,i}^{\text{o}(k)})}\right)},
\end{align*} 
\begin{align*}
b_{ig}^{(k)}&\colonequals\E({1}/{W_{ig}}\mid \vecy_i^{\text{o}},z_{ig}=1;\matTheta^{(k)}) \\
&= -\frac{2\lambda_g^{(k)}-p_i^{\text{o}}}{\omega_g^{(k)}+\delta(\vecy_i^{\text{o}},\vecmu_{g,i}^{\text{o}(k)}\mid\matsig_{g,i}^{\text{oo}(k)})}+\sqrt{\frac{\omega_g^{(k)}+(\vecbeta_{g,i}^{\text{o}(k)})^{\intercal}(\matsig_{g,i}^{\text{oo}(k)})\inv\vecbeta_{g,i}^{\text{o}(k)}}{\omega_g^{(k)}+\delta(\vecy_i^{\text{o}},\vecmu_{g,i}^{\text{o}(k)}\mid\matsig_{g,i}^{\text{oo}(k)})}}\\
&\qquad\times\frac{K_{\lambda_g^{(k)}-\frac{p_i^0}{2}+1}\left(\sqrt{(\omega_g^{(k)}+\delta(\vecy_i^{\text{o}},\vecmu_{g,i}^{\text{o}(k)}\mid\matsig_{g,i}^{\text{oo}(k)}))(\omega_g^{(k)}+(\vecbeta_{g,i}^{\text{o}(k)})^{\intercal}(\matsig_{g,i}^{\text{oo}(k)})\inv\vecbeta_{g,i}^{\text{o}(k)})}\right)}{K_{\lambda_g^{(k)}-\frac{p_i^0}{2}}\left(\sqrt{(\omega_g^{(k)}+\delta(\vecy_i^{\text{o}},\vecmu_{g,i}^{\text{o}(k)}\mid\matsig_{g,i}^{\text{oo}(k)}))(\omega_g^{(k)}+(\vecbeta_{g,i}^{\text{o}(k)})^{\intercal}(\matsig_{g,i}^{\text{oo}(k)})\inv\vecbeta_{g,i}^{\text{o}(k)})}\right)},\\
c_{ig}^{(k)}&\colonequals\E(\log W_{ig}\mid \vecy_i^{\text{o}},z_{ig}=1;\matTheta^{(k)})=\log\left( \sqrt{\frac{\omega_g^{(k)}+\delta(\vecy_i^{\text{o}},\vecmu_{g,i}^{\text{o}(k)}\mid\matsig_{g,i}^{\text{oo}(k)})}{\omega_g^{(k)}+(\vecbeta_{g,i}^{\text{o}(k)})^{\intercal}(\matsig_{g,i}^{\text{oo}(k)})\inv\vecbeta_{g,i}^{\text{o}(k)}}} \right)\\
&+\left.\frac{\partial}{\partial t}\log\left\{K_{t}\left(\sqrt{(\omega_g^{(k)}+\delta(\vecy_i^{\text{o}},\vecmu_{g,i}^{\text{o}(k)}\mid\matsig_{g,i}^{\text{oo}(k)}))(\omega_g^{(k)}+(\vecbeta_{g,i}^{\text{o}(k)})^{\intercal}(\matsig_{g,i}^{\text{oo}(k)})\inv\vecbeta_{g,i}^{\text{o}(k)})}\right)\right\}\right\rvert_{t=(\lambda_g^{(k)}-\frac{p_i^{\text{o}}}{2})}.
\end{align*} 
For convenience, we use the following notation analogous to \citet{browne15}: $n_g^{(k)} = \isum \hat{z}_{ig}^{(k)}$, $\bar{a}_g^{(k)} = 1/n_g^{(k)}\isum \hat{z}_{ig}^{(k)}a_{ig}^{(k)}$, $\bar{b}_g^{(k)} = 1/n_g^{(k)}\isum \hat{z}_{ig}^{(k)}b_{ig}^{(k)}$, and $\bar{c}_g^{(k)} = 1/n_g^{(k)}\isum \hat{z}_{ig}^{(k)}c_{ig}^{(k)}$. For the actual missing data $\vecY^{\text{m}}$, we will also need the following expectations:
\begin{align*}
\label{eqn:expymiss}
\hat{\vecy}_{ig}^{\text{m}(k)}&\colonequals\E(\vecY_i^{\text{m}}\mid\vecy_i^{\text{o}},z_{ig}=1)=\vecmu_{g,i}^{\text{m}\mid \text{o}(k)}+a_{ig}^{(k)}\vecbeta_{g,i}^{\text{m}\mid \text{o}(k)},\\
\tilde{\vecy}_{ig}^{\text{m}(k)}&\colonequals\E(({1}/{W_{i}})\vecY_i^{\text{m}}\mid\vecy_i^{\text{o}},z_{ig}=1)=b_{ig}^{(k)}\vecmu_{g,i}^{\text{m}\mid \text{o}(k)}+\vecbeta_{g,i}^{\text{m}\mid \text{o}(k)},\\
\tilde{\tilde{\vecy}}_{ig}^{\text{m}(k)}&\colonequals\E(({1}/{W_{i}})\vecY_i^{\text{m}}\vecY_i^{\text{m}\intercal}
\mid\vecy_i^{\text{o}},z_{ig}=1)=\matsig_{g,i}^{\text{m}\mid \text{o}(k)}+b_{ig}^{(k)}\vecmu_{g,i}^{\text{m}\mid \text{o}(k)}(\vecmu_{g,i}^{\text{m}\mid \text{o}(k)})^{\intercal}\\
&\quad+\vecmu_{g,i}^{\text{m}\mid \text{o}(k)}(\vecbeta_{g,i}^{\text{m}\mid \text{o}(k)})^{\intercal}+\vecbeta_{g,i}^{\text{m}\mid \text{o}(k)}(\vecmu_{g,i}^{\text{m}\mid \text{o}(k)})^{\intercal}+a_{ig}^{(k)}\vecbeta_{g,i}^{\text{m}\mid \text{o}(k)}(\vecbeta_{g,i}^{\text{m}\mid \text{o}(k)})^{\intercal}.
\end{align*} 

On the $k$-th iteration of the M-step, the expected value of the complete data log-likelihood is maximized to get the updates for the parameter estimates as follows:
\begin{align*}
\pi_g^{(k+1)} &= \frac{n_g^{(k)}}{n},\\
{\vecmu}_g^{(k+1)} &= \frac{1}{\isum\hat{z}_{ig}^{(k)}(\bar{a}_g^{(k)}b_{ig}^{(k)}-1)}\isum \hat{z}_{ig}^{(k)}\begin{pmatrix}(\bar{a}_g^{(k)}b_{ig}^{(k)}-1)\vecy_i^{\text{o}}\\\bar{a}_g^{(k)}\tilde{\vecy}_{ig}^{\text{m}(k)}-\hat{\vecy}_{ig}^{\text{m}(k)}\end{pmatrix},\\
{\vecbeta}_g^{(k+1)}&=\frac{1}{\isum\hat{z}_{ig}^{(k)}(\bar{a}_g^{(k)}b_{ig}^{(k)}-1)}\isum \hat{z}_{ig}^{(k)}\begin{pmatrix}(\bar{b}_g^{(k)}-b_{ig}^{(k)})\vecy_i^{\text{o}}\\\bar{b}_g^{(k)}\hat{\vecy}_{ig}^{\text{m}(k)}-\tilde{\vecy}_{ig}^{\text{m}(k)}\end{pmatrix},\\
{\matsig}_g^{(k+1)}&=\frac{1}{n_g^{(k)}}\isum\hat{z}_{ig}^{(k)}{\matsig}_{ig}^{(k+1)}-(\bar{\vecy}_g-{\vecmu}_g^{(k+1)}){\vecbeta}_g^{(k+1)\intercal}-{\vecbeta}_g^{(k+1)}(\bar{\vecy}_g-{\vecmu}_g^{(k+1)})^{\intercal}+\bar{a}_g^{(k)}{\vecbeta}_g^{(k+1)}{\vecbeta}_g^{(k+1)\intercal},
\end{align*}
where 
\begin{align*}
\bar{\vecy}_g&=\frac{1}{n_g^{(k+1)}}\isum\hat{z}_{ig}^{(k+1)}\begin{pmatrix}\vecy_i^{\text{o}}\\\hat{\vecy}_{ig}^{\text{m}(k+1)}\end{pmatrix},\\[+4pt]
{\matsig}_{ig}^{(k+1)} &= \begin{pmatrix}b_{ig}^{(k)}(\vecy_i^{\text{o}}-\vecmu_g^{\text{o}(k+1)})(\vecy_i^{\text{o}}-\vecmu_g^{\text{o}(k+1)})^{\intercal}&(\vecy_i^{\text{o}}-\hat{\vecmu}_g^{\text{o}(k+1)})(\tilde{\vecy}_{ig}^{\text{m}(k)}-b_{ig}^{(k)}\hat{\vecmu}_g^{\text{m}(k+1)})^{\intercal}\\(\tilde{\vecy}_{ig}^{\text{m}(k)}-b_{ig}^{(k)}\hat{\vecmu}_g^{\text{m}(k+1)})(\vecy_i^{\text{o}}-\vecmu_g^{\text{o}(k+1)})^{\intercal}&\mathbf{k}_{ig}^{\text{m}(k+1)}\end{pmatrix},
\end{align*} and
$$\mathbf{k}_{ig}^{\text{m}(k+1)}=\tilde{\tilde{\vecy}}_{ig}^{\text{m}(k)}-\tilde{\vecy}_{ig}^{\text{m}(k)}\hat{\vecmu}_g^{\text{m}(k+1)\intercal}-\hat{\vecmu}_g^{\text{m}(k+1)}\tilde{\vecy}_i^{\text{m}(k)\intercal}+b_{ig}^{(k)}\hat{\vecmu}_g^{\text{m}(k+1)}\hat{\vecmu}_g^{\text{m}(k+1)\intercal}.$$

Finally, the estimates of $\lambda_g^{(k+1)}$ and $\omega_g^{(k+1)}$ are given as solutions to maximize the function
\begin{equation}
\nonumber
q_g(\lambda_g,\omega_g) = -\log(K_{\lambda_g}(\omega_g))+(\lambda_g-1)\bar{c}_g-\frac{\omega_g}{2}(\bar{a}_g+\bar{b}_g),
\end{equation}
 and the associated updates are
 \begin{align*}
 {\lambda}_g^{(k+1)} &= \bar{c}_g^{(k)}{\lambda}_g^{(k)}\left[\frac{\partial}{\partial{\lambda}_g^{(k)}}\log\left(K_{{\lambda}_g^{(k)}}({\omega}_g^{(k)})\right)\right]\inv,\\
 {\omega}_g^{(k+1)} &={\omega}_g^{(k)}-\left[\frac{\partial}{\partial{\omega}_g^{(k)}}q_g({\lambda}_g^{(k+1)},{\omega}_g^{(k)})\right]\left[\frac{\partial^2}{\partial{\omega}_g^{2(k)}}q_g({\lambda}_g^{(k+1)},{\omega}_g^{(k)})\right]\inv.
 \end{align*}

The family of MGHD mixture models, with 14 parsimonious eigen-decomposed scaled covariance matrices corresponding to the famous GPCM family of models is proposed (see Appendix~\ref{sec:gpcm} for a brief discussion, including nomenclature). Details on the MST with incomplete data are analogous to the MGHD with incomplete data and are provided in Appendix~\ref{sec:MST}.

\section{Notes on Implementation}\label{sec:notes}
\subsection{Initial values}
It is well known that the EM algorithm can be heavily dependent on the initial values; indeed, good initial values of parameter estimates may speed up convergence. In this study, the following procedure for automatically generating initial values is used, unless otherwise specified.
\begin{itemize}
\item Fill in the missing values based on the mean imputation method.
\item Perform $k$-means clustering and use the resulting clustering membership to initialize the \textit{a~posteriori} probability $\hat{z}_{ig}^{(0)}$. Accordingly, the initial values for the model parameters are then given by: 
\begin{align*}
\hat{\pi}_g^{(0)}&=\frac{\isum \hat{z}_{ig}^{(0)}}{n},&\hat{\vecmu}_g^{(0)}&=\frac{\isum\hat{z}_{ig}^{(0)}\vecy_i}{\isum \hat{z}_{ig}^{(0)}},&\hat{\matsig}_g^{(0)}&=\frac{\isum\hat{z}_{ig}^{(0)}(\vecy_i-\hat{\vecmu}_g^{(0)})(\vecy_i-\hat{\vecmu}_g^{(0)})^{\intercal}}{\isum \hat{z}_{ig}^{(0)}}.
\end{align*}
\item Set the skewness parameter $\vecbeta_g^{(0)}$ to be close to zero for symmetric data.
\item When applicable, we set $\omega_g^{(0)}=1$ and $\lambda_g^{(0)}=-1/2$ for the index and concentration parameters, which represents a special case of GHD (i.e., normal-inverse Gaussian) distribution, or set $v_g^{(0)}=50$ for the near-normality assumption.
\end{itemize}

To enhance the computational efficiency of the EM algorithm, we update the parameters per missing pattern instead of per individual. We suggest rearranging $\vecY$ according to unique patterns of the missing  data. The procedure can be implemented as follows:
\begin{itemize}
\item Build a binary $n$ by $p$ indicator matrix $\mathbf{R}=[r_{ij}]$, with each entry $r_{ij}=1$ if $\vecY_{ij}$ is missing and $r_{ij}=0$ otherwise;
\item Find all unique missing patterns; and
\item Update parameters per missing pattern instead of per individual.
\end{itemize}
 
\subsection{Model Selection and Stopping Criterion}
In general, the number of mixture components $G$ is not known \textit{a priori}, and needs to be estimated from the data. Two widely used model selection techniques are the Bayesian information criterion \citep[BIC;][]{schwarz78} and the integrated completed likelihood \citep[ICL;][]{biernacki00}, which are given respectively by
\begin{equation}
\nonumber
\text{BIC} = 2l(\vecy,\hat{\matTheta}) - \rho \log(n) \quad \text{and} \quad  \text{ICL} \approx \text{BIC} + 2\isum \gsum \text{MAP}\left\{\hat{z}_{ig}\right\}\log(\hat{z}_{ig}),
\end{equation}
where $l(\hat{\matTheta})$ is the maximized log-likelihood evaluated at the maximum likelihood estimate $\hat{\matTheta}$, $\rho$ is the number of free parameters, $n$ is the number of observations, $\hat{z}_{ig}$ represents the estimated \textit{a~posteriori} probability that $\vecy_i$ arises from the $g$th component, and \text{MAP} denotes the maximum \emph{a~posteriori} probability such that $\text{MAP}\left\{\hat{z}_{ig}\right\}=1$ if $\text{max}_g\left\{\hat{z}_{ig}\right\}$ occurs in the $g$th component and $\text{MAP}\left\{\hat{z}_{ig}\right\}=0$ otherwise. 
The bigger the BIC or ICL value, the better the fitted model.

The EM algorithm can be stopped iterations after the maximum number of iterations, or when the Aitken stopping criterion~\citep{aitken26}  is satisfied. The Aitken acceleration at iteration $k$ is
\begin{equation}
\nonumber
a^{(k)} = \frac{l^{(k+1)}-l^{(k)}}{l^{(k)}-l^{(k-1)}},
\end{equation}
where $l^{(k)}$ is the log-likelihood at iterations $k$. This yields an asymptotic estimate of the log-likelihood at iteration $k+1$: $$l_{\infty}^{(k+1)} = l^{(k)} + \frac{1}{1-a^{(k)}}(l^{(k+1)}-l^{(k)})$$ \citep{bohning94,lindsay95}, and the EM algorithm is stopped when $l_{\infty}^{(k+1)}-l^{(k)} < \epsilon$, provided this difference is positive \citep{mcnicholas10a}.

\section{Numerical Examples}\label{sec:example}
Studies based on both simulated and real datasets are used to compare the clustering performance of the proposed approach. 
Our proposed family of models for incomplete data is compared to multivariate \textit{t} mixture with ML estimation in the presence of missing values (M\textit{t}). BIC is used to select the model; models with higher values of BIC are preferable. The adjusted Rand index \citep[ARI;][]{hubert85} is used to compare predicted classifications to true classes when applicable. The Rand index \citep{rand71} is the ratio of pairwise agreements to total pairs, and the ARI corrects the Rand index to account for chance agreement. The ARI has expected value 0 under random classification and takes the value 1 for perfect class agreement. A detailed discussion of the ARI, and arguments in favour of its use, are given by \cite{steinley04}.
 
\subsection{Simulation Studies}
The simulated datasets are each two-component mixtures: a mixture of Gaussian distributions (GMM) with a general VEE covariance structure, a mixture of skew-t distributions (MST) with a diagonal VEI covariance structure, and a mixture of generalized hyperbolic distributions (MGHD) with a general VEE covariance structure.~The GMM datasets are generated via the \textsf{R} function {\tt{rmvnorm}} from the {\tt mvtnorm} package for \textsf{R}, and the MST and MGHD datasets are generated using {\sf R} code based on the stochastic representations in \eqref{eqn:skewtstochastic} and \eqref{eqn:ghd2}, respectively. 

For each mixture component, $n_g=200$ two-dimensional vectors $\vecy_i$ are generated. The presumed parameters of $\matsig_g$ ($g=1,2$) for the VEE and VEI models are the same as those considered in \citet{celeux95} and \citet{linlearning14}. Each mixture component is centred on a different point giving well-separated and overlapping mixtures. Where applicable, the skewness parameters are $\vecbeta_1=(1,1)^{\intercal}$ and $\vecbeta_2=(-1,-1)^{\intercal}$, the degrees of freedoms for the MST is $v_1=7$ and $v_2=5$, and the values of other parameters for the MGHD are $\omega_1=\omega_2=6$ and $\lambda_1=-{1}/{2}$ and $\lambda_2=1$. 

The datasets considered in the simulation studies are summarized in Table~\ref{table:Sim} and examples are plotted in Figure~\ref{figure:Sim}. The datasets are overlapping, making this a relatively difficult clustering scenario even when the datasets are complete.
\begin{table}[ht]
\centering
\caption{Summary of simulated datasets.}
\label{table:Sim}
\begin{tabular}{lrrr}
  \hline
 Dataset & Distribution & Covariance structure $(\matsig_g)$ & Separation between components \\ 
  \hline
   Sim1 & MGHD & VEE & Well-separated \\ 
   Sim2 & MGHD & VEE & Overlapping \\ 
   Sim3 & MST & VEI & Well-separated \\ 
   Sim4 & MST & VEI & Overlapping \\ 
   Sim5 & GMM & VEE & Well-separated \\ 
   Sim6 & GMM & VEE & Overlapping \\ 
   \hline
\end{tabular}
\end{table}
\begin{figure}[!ht]
	\centering
	\includegraphics[width=0.75\textwidth]{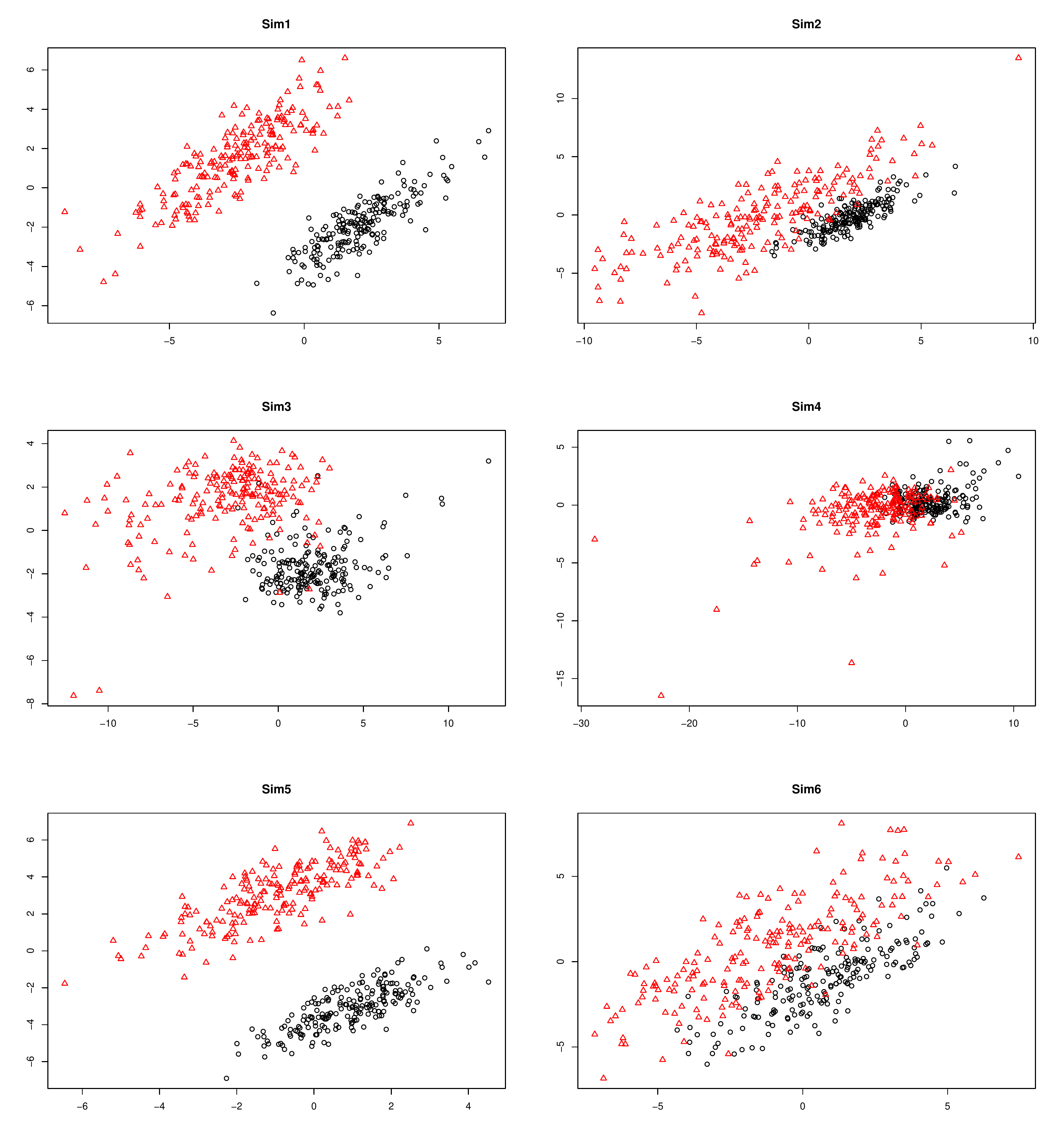}
	\vspace{-0.1in}
	\caption{Exemplar scatter plots for simulated datasets, where colour and plotting symbol represent true labels (component membership).}
	\label{figure:Sim}
\end{figure}

Artificial missing datasets are simulated by removing $n\times r$ elements from each column of the simulated samples through two different MAR patterns and the MCAR mechanism under three missing rates --- $r=0.05$ (low), $r=0.15$ (moderate), and $r=0.3$ (high) --- while maintaining the condition that each observation has at least one observed attribute. For the MAR mechanism, data points in the first column are sorted in descending order. Column $2$ is then divided into four equal blocks and, for each block, a specified number of elements (see Table~\ref{tab:pat}) are removed at random. When $p=1$, the second column is used.
\begin{table*}[h]
	\caption{Number of missing observations for each pattern.}
	\centering
	\begin{tabular*}{1.0\textwidth}{@{\extracolsep{\fill}}llll}
	\hline
	$r$&Pattern 1&Pattern 2\\
	\hline
$5\%$&(10,3,6,1)&(1,6,3,10)\\
$15\%$&(30,9,18,3)&(3,18,9,30)\\

$30\%$&(60,18,36,6)&(6,36,18,60)\\
\hline
	\end{tabular*}
	\label{tab:pat}
\end{table*}

First, we examine the ability of our proposed model to recover underlying parameters when the number of components and the covariance structure are correctly specified. These experiments comprise 100 replications per combination of missing pattern and missingness rate. The means of the parameter estimates with their associated standard deviations and bias are summarized in Table~\ref{sim1parresult} and ~\ref{sim3parresult} (Appendix~\ref{app:tables}). The means of  most parameter estimates are close to the true values with small standard deviations when $r=0.05$. The standard deviations increase as the missing rate increases, while at the same time, the average ARI slightly decreases. The means of estimated $\lambda_1$ and $\lambda_2$ in Sim1 are quite far from the true value because we obtain those estimates using an approximation to the Bessel function. In addition, there is no significant difference among the three missing patterns. Therefore, we use MCAR in the rest of the data examples. 

As another illustration, we explore the flexibility of the MGHD model for incomplete data and study the performance of the BIC for model selection. As mentioned in the introduction, the GHD is a flexible distribution with skewness, concentration, and index parameters. 
We compute the average ARI for the parsimonious MGHD and MST models introduced here as well as M\textit{t} under the circumstances of unknown clusters ($G=1,\ldots,4$). The detailed results are summarized in Table~\ref{table:SimNG} (Appendix~\ref{app:tables}). From Table~\ref{table:SimNG}, we observe the following:
\begin{itemize}
\item The average ARI decreases as the missing rate rises. As expected, overlapping components typically have lower ARI than the well-separated components. In addition, the average ARI considerably decreases when the missing rate reaches 30\% $(r=0.30)$ for Sim2, Sim4 and Sim6.

\item Our proposed parsimonious MGHD models for incomplete data perform significantly better than M\textit{t}. The family of MGHD models generally yields much higher ARI than its competitor parsimonious MST for incomplete data when the datasets are generated from a generalized hyperbolic distribution.

\item The BIC always finds the true number of clusters when using the MGHD for incomplete data, but tends to overestimate the number of clusters when using the MST  or M\textit{t} for incomplete data for datasets with overlapping mixtures. 

\item The BIC prefers MGHD over M\textit{t} in Sim5 and Sim6 where the data is generated from GMMs. We find that the samples are not necessarily symmetric, particularly with missing values. Figure~\ref{figure:Sim5T} and \ref{figure:Sim6T} show exemplar scatter plots for data from Sim5 and Sim6 for $r=0.10$.  The M\textit{t} tends to overestimate the number of clusters, hence, has a lower averaged BIC.

\end{itemize}

\begin{figure}[!ht]
	\centering
	\includegraphics[width=0.60\textwidth]{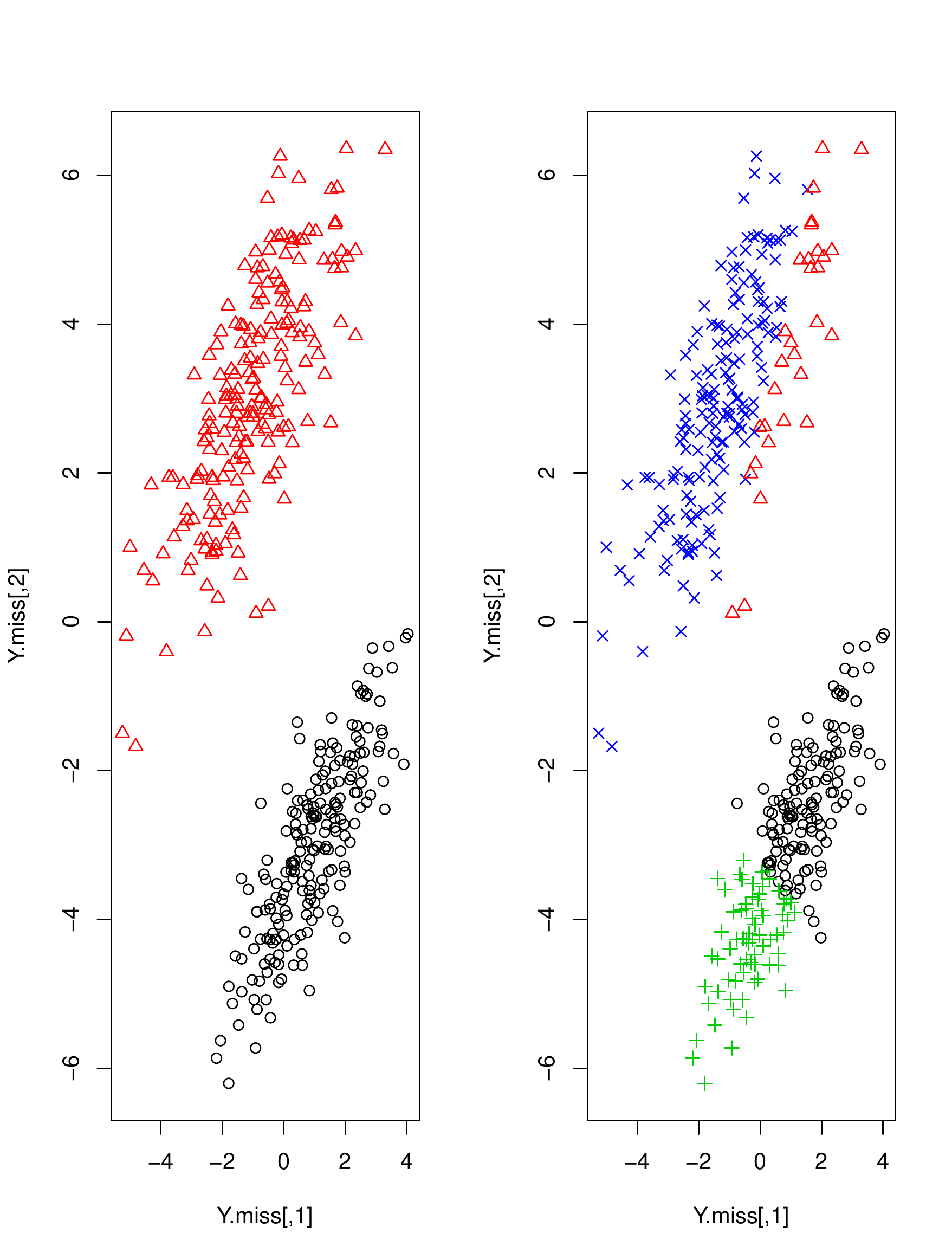}
	\vspace{-0.1in}
	\caption{Exemplar scatter plots for Sim5, with true labels (left) and clustering results from the best M\textit{t} models (right), where colour and plotting symbol represent true (left) or predicted (right) class.}
	\label{figure:Sim5T}
\end{figure}

\begin{figure}[!ht]
	\centering
	\includegraphics[width=0.60\textwidth]{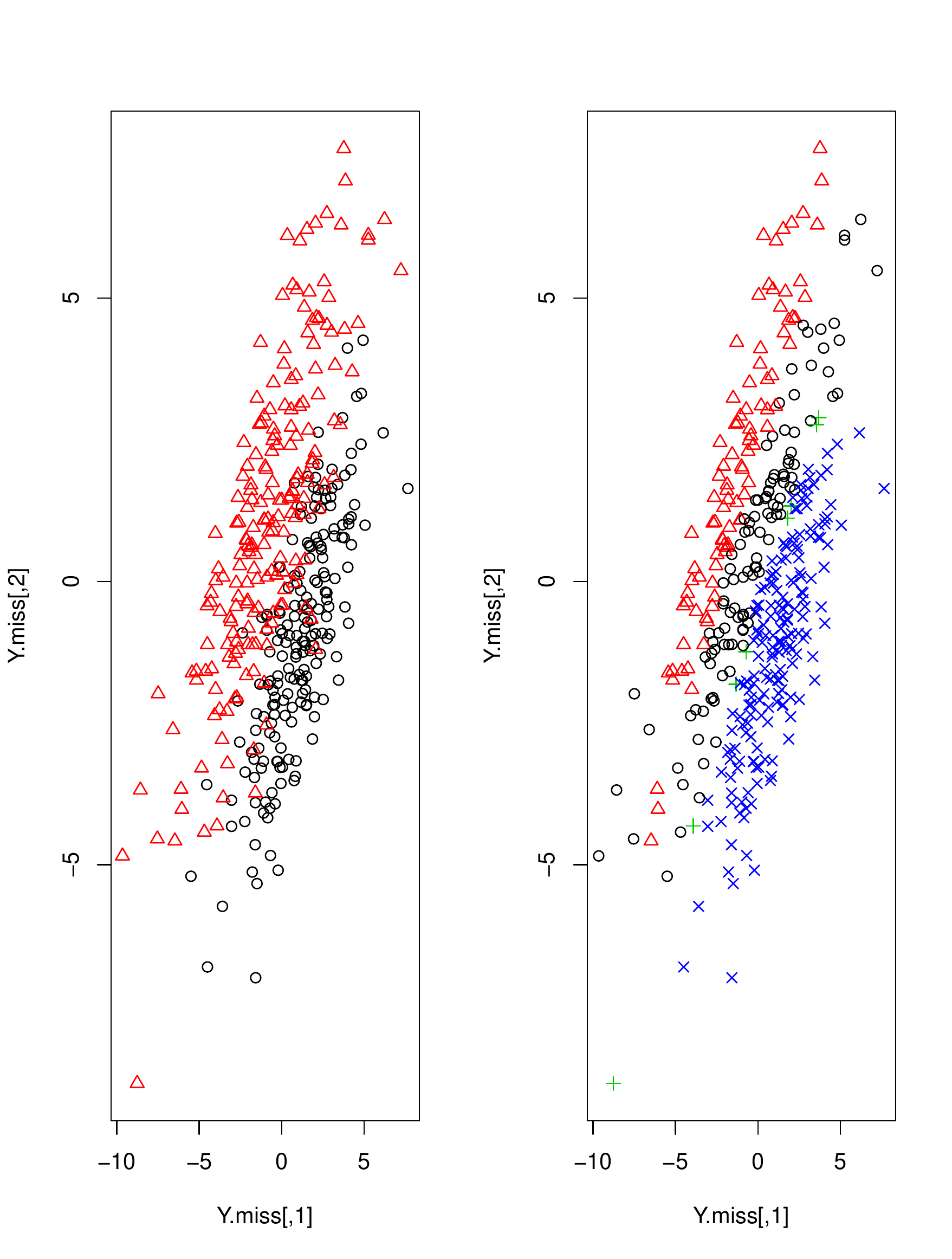}
	\vspace{-0.1in}
	\caption{Exemplar scatter plots for Sim6, with true labels (left) and clustering results from the best M\textit{t} models (right), where colour and plotting symbol represent true (left) or predicted (right) class.}
	\label{figure:Sim6T}
\end{figure}

\subsection{Breast Cancer Diagnostic Dataset}
The breast cancer diagnostic data consists of ten real-valued features on 569 cases of breast tumours -- 357 benign and 212 malignant. The mean, standard error, and ``worst'' or largest of these features were computed for each image, resulting in 30 attributes. This dataset is complete, so for illustration purposes we consider levels of missing data $r=0.05$ and $r=0.15$ by deleting observations through an MCAR mechanism while maintaining the condition that each observation has at least one observed attribute. The dataset is scaled prior to analysis.

The family of MGHD, MST and M\textit{t} models were fitted to these data for $G=1,\ldots,4$. We randomly assign each observation to one of the G groups and start with 20 random initializations of the algorithm, selecting the model with the maximum likelihood values. The key statistics of the best models for MGHD, MST and M\textit{t} are shown in Table~\ref{table:bcdcomp.}. The results of this analysis show that the parsimonious MGHD outperforms the other models for all levels of missing data.

\begin{table}[!htb]
\centering
\vspace{-0.1cm}
\caption{A comparison of averaged BIC, ARI and the number of times (nt) when $G=2$ is chosen among MGHD, MST, and M\textit{t} models on the tumour dataset with $G=1,\ldots,4$.}
\vspace{-0.12in}
\scalebox{1}{
\begin{tabular*}{1\textwidth}{@{\extracolsep{\fill}}lcccccc}
\hline
&\multicolumn{3}{c}{$r=0.05$}&\multicolumn{3}{c}{$r=0.15$}\\
\cline{2-4}\cline{5-7}
&Avg.BIC&Avg.ARI&nt&Avg.BIC&Avg.ARI&nt\\ [0.5ex] 
\hline
MGHD&$12145$&$0.65$&$18$&$9654$&$0.58$&$16$\\
MST&$12661$&$0.55$&$15$&$10574$&$0.56$&$16$\\
M\textit{t}&$13605$&$0.47$&$10$&$11605$&$0.36$&$10$\\
\hline
\end{tabular*}}
\label{table:bcdcomp.}
\end{table}

\subsection{Pima Indians Diabetes Data}
Data on the diabetes status of 768 patients is obtained from the UCI Machine Learning data repository. The data include information on eight attributes, in which the attribute of number of times pregnant is treated as continuous variable because its range is from 0~to~14. These data are a popular benchmark dataset for clustering for truly missing values, as 376 of the observations have at least one attribute missing. The data are overlapping and the numerous missing observations make clustering difficult. The detailed description of the attributes and their associated missing rates are summarized in Table~\ref{table:pima}. The dataset features 268 patients with a diabetes diagnosis and 500 without, and these are treated as two clusters.  Again, this dataset is scaled prior to the analysis.
 \begin{table}[ht]
\centering
\caption{A description of Pima Indian diabetes dataset.} 
\label{table:pima}
\begin{tabular}{lrrr}
  \hline
   & No.\ missing values & Sample mean & Sample std.\ dev.\\ 
  \hline
   Number of times pregnant &   0 &   3.85 &   3.37 \\ 
   Plasma glucose concentration  &   5 & 120.89 &  31.97 \\ 
   Diastolic blood pressure (mm Hg) &  35 &  69.11 &  19.36 \\ 
   Triceps skin fold thickness (mm) & 227 &  20.54 &  15.95 \\ 
   2-hour serum insulin(mu U/mL) & 374 &  79.80 & 115.24 \\ 
   Body mass index &  11 &  31.99 &   7.88 \\ 
   Diabetes pedigree function &   0 &   0.47 &   0.33 \\ 
   Age (years) &   0 &  33.24 &  11.76 \\ 
   \hline
\end{tabular}
\end{table}

Because there are two known clusters,  we fix $G=2$ and compare the BIC and ICL values for 14 covariance structures of our proposed parsimonious MGHD and MST models. The clustering results are summarized in Table~\ref{table:pimadata}. \citet{linlearning14} perform the M\textit{t} and matches the true cluster labels with 66.7\% accuracy.  Compared to~\citet{linlearning14}, our proposed parsimonious MGHD model for incomplete data gives a higher accuracy rate (69.11\%).
\begin{table}[ht]
\centering
\caption{The BIC, ICL, selected $\matsig_g$ and the correct classification rate for our proposed approaches for clustering on the Pima Indian diabetes dataset.}
\label{table:pimadata}
\begin{tabular*}{0.75\textwidth}{@{\extracolsep{\fill}}lrrrr}
  \hline
  &$\matsig_g$&BIC&ICL&Accuracy\\
  \hline
   MGHD&EVE&$-14016.95$&$-14053.61$& 69.11\% \\ 
   MST&VVI& $-14109.1$&$-14186.1$& 62.37\% \\
   \hline
\end{tabular*}
\end{table}
The best model  is the two-component MGHD model and $\matsig_g$=EVE. Group 1 consists mainly of the non-diabetic patients and Group 2 consists mainly of the diabetic patients. We then fit the best model with 100 random initializations; Table~\ref{table:pimaparest} shows the key parameter estimates for this model as well as the corresponding standard errors. The standard errors of the model parameters have been calculated using the bootstrap method described in \citet{efron86}. The estimates for $\vecmu_g+\vecbeta_g$ are quite similar to the parameter estimates presented in \citet{wang15}. The estimates for the skewness parameters indicate the presence of skewness in most of the variables. 
\begin{table}[htbp]
	\centering
	\caption{Summary of key model parameter estimates (standard errors) for the best chosen model (i.e., MGHD with $\matsig_g=$ EVE) for the Pima Indian diabetes dataset.}
		\begin{tabular*}{0.65\textwidth}{@{\extracolsep{\fill}}lrr}
		  \hline	
		Parameter&$g=1$& $g=2$ \\
		\hline
		$\mu_{1g}$&$-0.80$ (0.11) &2.98 (1.78)\\
		$\mu_{2g}$&$-0.97$ (0.22) &1.35 (4.01)\\
		$\mu_{3g}$ & $-0.69$ (0.14) & 1.10 (2.65)\\
		$\mu_{4g}$& 0.15 (0.08) & $-0.50$ (4.59)\\
		$\mu_{5g}$&$-1.26$ (1.73) &0.18 (0.25)\\
		$\mu_{6g}$&$-0.66$ (0.07)&0.57 (0.78)\\
		$\mu_{7g}$&$-0.74$ (0.12)&$-2.67$ (8.41)\\
		$\mu_{8g}$ &$-1.20$ (0.31)& $-2.01$ (2.04)\\
		$\beta_{1g}$ &0.57 (0.05) & $-2.18$ (1.79)\\
		$\beta_{2g}$ &0.77 (0.47) & $-0.92$ (0.25)\\
		$\beta_{3g}$&0.54 (0.40)& $-0.78$ (1.18)\\
		$\beta_{4g}$ &0.53 (0.31) & 0.58 (0.38)\\
		$\beta_{5g}$&0.11 (0.13) & 0.11 (0.32)\\
		$\beta_{6g}$ &0.57 (0.16) &$-0.37$ (0.51)\\
		$\beta_{7g}$ &0.63 (0.18) &1.27 (0.46)\\
		$\beta_{8g}$ &0.87 (0.16) & 2.91 (1.85)\\		
		$\omega_g$ & 2.39 (1.81) &14.18 (6.83)\\
		$\lambda_g$ &0.02 (0.34)&$-3.18$ (4.60)\\
		$\pi_g$&0.71 (0.09)& 0.29 (0.10) \\
		\hline
	\end{tabular*}
	\label{table:pimaparest}
\end{table}

\section{Discussion}\label{sec:discussion}
Approaches for clustering incomplete data where clusters may be heavy tailed and/or asymmetric is introduced, based on MGHD and MST. There approaches were further extended to parsimonious families of MGHD and MST models via eigen-decomposition of the component scale matrices. The BIC and ICL were used for model selection. It is well known that the BIC can tend to overestimate the number of clusters in practice; however, the results presented herein show that this overestimation can sometimes be mitigated via a more flexible component density such as the MGHD. An EM algorithm was developed to fit the MGHD and MST models to incomplete data, and later implemented in {\sf R}. It is worth mentioning that our approaches are also applicable in situations with no missing data; and so we have MGHD and MST analogues of the models of \cite{celeux95}. Our MGHD and MST models were applied to real and simulated heterogeneous datasets for clustering in the presence of missing values, and the PMGHD family performed favourably when compared to the PMST family as well as the MGHD and MST approaches with mean imputation.

In the present work, the missing data mechanism is assumed to be MAR. Future work will focus on a departure from this assumption. As a starting point, the behaviour of parameter estimates for models considered herein when we depart from the MAR assumption will be studied.
Although we demonstrated the PMGHD and PMST approaches for clustering, they also can be applied for semi-supervised classification, discriminant analysis, and density estimation; furthermore, they could be used within the fractionally-supervised paradigm \citep{vrbik15}. Furthermore, Bayesian analysis via a Gibbs sampler is another popular approach to handle missing data in multivariate datasets \citep[e.g.,][]{lin09a}, so a fully Bayesian treatment will be considered as an alternative to the EM algorithm for parameter estimation. Finally, it will also be interesting to generalize all existing approaches to developing mixture of generalized hyperbolic factor analyzer models \citep{tortora16}, mixtures with hypercube contours \citep{franczak15}, and mixtures of multiple scaled generalized hyperbolic distributions for incomplete data \citep{tortora14}.

{\small\paragraph{Acknowledgements} This work was supported by an Ontario Graduate Scholarship (Wei), an Early Researcher Award from the Government of Ontario (McNicholas), and the Canada Research Chairs program (McNicholas).}


\appendix
{\small\section{GPCM Family}\label{sec:gpcm}
\cite{banfield93} consider an eigen-decomposition of the component scale matrices (which is equivalent to the component covariance matrices for Gaussian mixtures), i.e.,
\begin{equation}\label{eqn:gpcm}
\matsig_g=\lambda_g\gam_g\del_g\gam_g',
\end{equation}
where $\lambda_g=\left|\matsig_g\right|^{1/p}$, $\gam_g$ is the matrix of eigenvectors of $\matsig_g$, and $\del_g$ is a diagonal matrix, such that $\left|\del_g\right|=1$, containing the normalized eigenvalues of $\matsig_g$ in decreasing order. Note that the columns of $\gam_g$ are ordered to correspond to the elements of $\del_g$. As \cite{banfield93} point out, the constituent elements of the decomposition in \eqref{eqn:gpcm} can be viewed in the context of the geometry of the component, where $\lambda_g$ represents the volume in $p$-space, $\del_g$ the shape, and $\gam_g$ the orientation.
By imposing constraints on the elements of the decomposed covariance structure in \eqref{eqn:gpcm}, \cite{celeux95} introduce a family of GPCMs (Table~\ref{tab:GPCM}). 
\begin{table}[ht]
\caption{The nomenclature and scale matrix structure for each member of the GPCM family.}\label{tab:GPCM}
\begin{tabular*}{\textwidth}{@{\extracolsep{\fill}}llllr}
\hline
{Nomenclature} & {Volume} & {Shape} & {Orientation} & $\matsig_g$\\
\hline
EII & Equal    & Spherical &             & $\lambda \ident$\\
VII & Variable & Spherical &             & $\lambda_g \ident$\\
EEI & Equal    & Equal     & Axis-Aligned & $\lambda \del$\\
VEI & Variable & Equal     & Axis-Aligned & $\lambda_g \del$\\
EVI & Equal    & Variable  & Axis-Aligned & $\lambda \del_g$\\
VVI & Variable & Variable  & Axis-Aligned & $\lambda_g \del_g$\\
EEE & Equal    & Equal     & Equal        & $\lambda\gam\del\gam'$\\
VEE & Variable & Equal     & Equal        & $\lambda_g\gam\del\gam'$\\
EVE & Equal    & Variable  & Equal        & $\lambda\gam\del_g\gam'$\\
EEV & Equal    & Equal     & Variable     & $\lambda\gam_g\del\gam_g'$\\
VVE & Variable & Variable  & Equal        & $\lambda_g\gam\del_g\gam'$\\
VEV & Variable & Equal     & Variable   & $\lambda_g\gam_g\del\gam_g'$\\
EVV & Equal    & Variable  & Variable   & $\lambda\gam_g\del_g\gam_g'$\\
VVV & Variable & Variable  & Variable   & $\lambda_g\gam_g\del_g\gam_g'$\\
\hline
\end{tabular*}
\end{table}

\section{Some Useful Matrix Computations}\label{app:matrixformula}
We here present some useful matrix computation results that are employed in the derivation of the conditional pdf of a partitioned generalized hyperbolic and multivariate skew-t random vector $\vecY$ in Propositions 3 and 6.

Consider a partitioned random vector $\vecY$ of $p$-dimension that follows the pdf as in \eqref{eqn:ghd3} with
\begin{align}
\label{eqn:parameterblock}
\vecY = \begin{pmatrix} \vecY_1 \\ \vecY_2 \end{pmatrix}&&\vecmu = \begin{pmatrix} \vecmu_1 \\ \vecmu_2 \end{pmatrix}&&\vecbeta = \begin{pmatrix} \vecbeta_1 \\ \vecbeta_2 \end{pmatrix}&&\matsig = \begin{pmatrix} \matsig_{11} & \matsig_{12}\\ \matsig_{21} &\matsig_{22} \end{pmatrix},
\end{align}
where $\vecY_1$ and $\vecY_2$ have dimensions $d_1$ and $d_2=p-d_1$, respectively. The mean, skewness and dispersion matrix are composed of blocks of appropriate dimensions as partitions of $\vecY$. Sometimes, it is more convenient to work with the inverse of dispersion matrix $\matsig\inv$:
\begin{align}
\label{eqn:siginverse}
\matsig\inv =\begin{pmatrix} (\matsig_{11}-\matsig_{12}\matsig_{22}\inv\matsig_{12}^{\intercal})\inv & -\matsig_{11}\inv\matsig_{12}(\matsig_{22}-\matsig_{12}^{\intercal}\matsig_{11}\inv\matsig_{12})\inv\\ -(\matsig_{22}-\matsig_{12}^{\intercal}\matsig_{11}\inv\matsig_{12})\inv\matsig_{12}^{\intercal}\matsig_{11}\inv &(\matsig_{22}-\matsig_{12}^{\intercal}\matsig_{11}\inv\matsig_{12})\inv\end{pmatrix}.
\end{align}
Furthermore, we have for the determinant of $\matsig$:
\begin{equation}
\label{eqn:detsigma}
\text{det}(\matsig) = \text{det}(\matsig_{11})\text{det}(\matsig_{22}-\matsig_{12}^{\intercal}\matsig_{11}\inv\matsig_{12}).
\end{equation}

\section{Outline of Proof of Proposition 3}
Here, we derive the conditional density of $\vecY_2$ given that $\vecY_1=\vecy_1$ if $\vecY_1$ and $\vecY_2$ are jointly generalized hyperbolic distributed, i.e., $\vecY \sim \text{GHD}_p(\lambda,\omega,\vecmu,\matsig,\vecbeta)$ with the partition in Appendix A. Although basic probability theory indicates that the conditional pdf is a ratio of the joint and marginal pdfs, the expression takes a very complicated form. The results from Appendix A are heavily used in the course of the derivations. The conditional density is given by
\begin{align*}
f_{\vecY_2 \mid \vecY_1}(\vecy_2 \mid \vecy_1) &= \frac{f_{\vecY_1,\vecY_2}(\vecy_1,\vecy_2)}{f_{\vecY_1}(\vecy_1)}\\
&=\frac{\left [\frac{\omega +\delta(\vecy,\vecmu \mid \matsig )}{\omega+\vecbeta^{\intercal}\matsig\inv\vecbeta}\right]^{\frac{\lambda-p/2}{2}}\frac{K_{\lambda-p/2}\left(\sqrt{(\omega +\delta(\vecy,\vecmu \mid \matsig))(\omega+\vecbeta^{\intercal}\matsig\inv\vecbeta)}\right)}{(2\pi)^{p/2} |\matsig|^{1/2}K_\lambda(\omega)\text{exp}\{-(\vecy-\vecmu)^{\intercal}\matsig\inv\vecbeta\}}}{\left [\frac{\omega +\delta(\vecy_1,\vecmu_1 \mid \matsig_{11} )}{\omega+\vecbeta_1^{\intercal}\matsig_{11}\inv\vecbeta_1}\right]^{\frac{\lambda-d_{1}/2}{2}}\frac{K_{\lambda-d_{1}/2}\left(\sqrt{(\omega +\delta(\vecy_1,\vecmu_1 \mid \matsig_{11}))(\omega+\vecbeta_1^{\intercal}\matsig_{11}\inv\vecbeta_1)}\right)}{(2\pi)^{d_{1}/2} |\matsig_{11}|^{1/2}K_\lambda(\omega)\text{exp}\{-(\vecy_1-\vecmu_1)^{\intercal}\matsig_{11}\inv\vecbeta_1\}}},
\end{align*}
where we combine \eqref{eqn:ghd3} and Proposition 2. For the moment, we focus on the linear form and quadratic form in which $\vecy$ enters the pdf in \eqref{eqn:ghd3}. Inserting the partition of $\vecY, \vecmu,\vecbeta$, and $\matsig$ in \eqref{eqn:parameterblock} and the inverse of dispersion matrix $\matsig\inv$ \eqref{eqn:siginverse} into the quadratic form yields
\begin{align}
\nonumber
\delta(\vecy,&\vecmu \mid \matsig ) = (\vecy - \vecmu)^{\intercal}\matsig\inv(\vecy - \vecmu)=\begin{pmatrix} (\vecy_1 - \vecmu_1)^{\intercal}&(\vecy_2 - \vecmu_2)^{\intercal}\end{pmatrix}\matsig\inv\begin{pmatrix} \vecy_1 - \vecmu_1\\\vecy_2 - \vecmu_2\end{pmatrix}\\
\nonumber
&=(\vecy_1 - \vecmu_1)^{\intercal}(\matsig_{11}-\matsig_{12}\matsig_{22}\inv\matsig_{12}^{\intercal})\inv(\vecy_1 - \vecmu_1)\\
\nonumber
&\quad-(\vecy_2 - \vecmu_2)^{\intercal}(\matsig_{22}-\matsig_{12}^{\intercal}\matsig_{11}\inv\matsig_{12})\inv\matsig_{12}^{\intercal}\matsig_{11}\inv(\vecy_1 - \vecmu_1)\\
\nonumber
&\quad-(\vecy_1 - \vecmu_1)^{\intercal}\matsig_{11}\inv\matsig_{12}(\matsig_{22}-\matsig_{12}^{\intercal}\matsig_{11}\inv\matsig_{12})\inv(\vecy_2 - \vecmu_2)\\
\nonumber
&\quad+(\vecy_2 - \vecmu_2)^{\intercal}(\matsig_{22}-\matsig_{12}^{\intercal}\matsig_{11}\inv\matsig_{12})\inv(\vecy_2 - \vecmu_2)\\
\nonumber
&=(\vecy_1 - \vecmu_1)^{\intercal}\matsig_{11}\inv(\vecy_1 - \vecmu_1)\\
\nonumber
&\quad+(\vecy_1 - \vecmu_1)^{\intercal}\matsig_{11}\inv\matsig_{12}(\matsig_{22}-\matsig_{12}^{\intercal}\matsig_{11}\inv\matsig_{12})\inv\matsig_{12}^{\intercal}\matsig_{11}\inv(\vecy_1 - \vecmu_1)\\
\nonumber
&\quad-(\vecy_2 - \vecmu_2)^{\intercal}(\matsig_{22}-\matsig_{12}^{\intercal}\matsig_{11}\inv\matsig_{12})\inv\matsig_{12}^{\intercal}\matsig_{11}\inv(\vecy_1 - \vecmu_1)\\
\nonumber
&\quad-(\vecy_1 - \vecmu_1)^{\intercal}\matsig_{11}\inv\matsig_{12}(\matsig_{22}-\matsig_{12}^{\intercal}\matsig_{11}\inv\matsig_{12})\inv(\vecy_2 - \vecmu_2)\\
\nonumber
&\quad+(\vecy_2 - \vecmu_2)^{\intercal}(\matsig_{22}-\matsig_{12}^{\intercal}\matsig_{11}\inv\matsig_{12})\inv(\vecy_2 - \vecmu_2)\\
\nonumber
&=(\vecy_1 - \vecmu_1)^{\intercal}\matsig_{11}\inv(\vecy_1 - \vecmu_1)\\
\nonumber
&\quad+(\vecy_2 - \vecmu_2-\matsig_{12}^{\intercal}\matsig_{11}\inv(\vecy_1 - \vecmu_1))^{\intercal}(\matsig_{22}-\matsig_{12}^{\intercal}\matsig_{11}\inv\matsig_{12})\inv(\vecy_2 - \vecmu_2-\matsig_{12}^{\intercal}\matsig_{11}\inv(\vecy_1 - \vecmu_1))\\
\label{eqn:quadraticy}
&=\delta(\vecy_1, \vecmu_1\mid \matsig_{11})+\delta(\vecy_2, \vecmu_{2\mid1}\mid \matsig_{2\mid1}),
\end{align}
where $\vecmu_{2\mid1}=\vecmu_2+\matsig_{12}^{\intercal}\matsig_{11}\inv(\vecy_1 - \vecmu_1)$ and $\matsig_{2\mid1}=(\matsig_{22}-\matsig_{12}^{\intercal}\matsig_{11}\inv\matsig_{12})\inv$.

Similarly, inserting into the linear form, following the same algebra as above, yields
\begin{align}
\nonumber
(&\vecy-\vecmu)^{\intercal}\matsig\inv\vecbeta = \begin{pmatrix} (\vecy_1 - \vecmu_1)^{\intercal}&(\vecy_2 - \vecmu_2)^{\intercal}\end{pmatrix}\matsig\inv\begin{pmatrix} \vecbeta_1\\\vecbeta_2 \end{pmatrix}\\
\nonumber
&=(\vecy_1 - \vecmu_1)^{\intercal}\matsig_{11}\inv\vecbeta_1+(\vecy_2 - \vecmu_2-\matsig_{12}^{\intercal}\matsig_{11}\inv(\vecy_1 - \vecmu_1))^{\intercal}(\matsig_{22}-\matsig_{12}^{\intercal}\matsig_{11}\inv\matsig_{12})\inv(\vecbeta_2-\matsig_{12}^{\intercal}\matsig_{11}\inv\vecbeta_1)\\
\label{eqn:lineary}
&=(\vecy_1 - \vecmu_1)^{\intercal}\matsig_{11}\inv\vecbeta_1+(\vecy_2-\vecmu_{2\mid1})^{\intercal}\matsig_{2\mid1}\inv\vecbeta_{2\mid1},
\end{align}
where $\vecmu_{2\mid1}$ and $\matsig_{2\mid1}$ are as described above, and $\vecbeta_{2\mid1} = \vecbeta_2-\matsig_{12}^{\intercal}\matsig_{11}\inv\vecbeta_1$.

Furthermore, we investigate the term $\vecbeta^{\intercal}\matsig\inv\vecbeta$, we obtain
\begin{align}
\nonumber
\vecbeta^{\intercal}\matsig\inv\vecbeta&=\begin{pmatrix} \vecbeta_1^{\intercal}&\vecbeta_2^{\intercal}\end{pmatrix}\matsig\inv\begin{pmatrix} \vecbeta_1\\\vecbeta_2 \end{pmatrix}\\
\nonumber
&=\vecbeta_1^{\intercal}\matsig_{11}\inv\vecbeta_1+(\vecbeta_2-\matsig_{12}^{\intercal}\matsig_{11}\inv\vecbeta_1)^{\intercal}(\matsig_{22}-\matsig_{12}^{\intercal}\matsig_{11}\inv\matsig_{12})\inv(\vecbeta_2-\matsig_{12}^{\intercal}\matsig_{11}\inv\vecbeta_1)\\
\label{eqn:quadraticbeta}
&=\vecbeta_1^{\intercal}\matsig_{11}\inv\vecbeta_1+\vecbeta_{2\mid1}^{\intercal}\matsig_{2\mid1}\vecbeta_{2\mid1}.
\end{align}

Finally,  we substitute \eqref{eqn:detsigma}, \eqref{eqn:quadraticy}, \eqref{eqn:lineary}, and \eqref{eqn:quadraticbeta}, and $p=d_1+d_2$ into the conditional density, and after some simple linear algebra, we obtain
\begin{align*}
f_{\vecY_2 \mid \vecY_1}&(\vecy_2 \mid \vecy_1) =\frac{\left(\frac{\omega+\delta(\vecy_1, \vecmu_1\mid \matsig_{11})+\delta(\vecy_2, \vecmu_{2\mid1}\mid \matsig_{2\mid1})}{\omega+\vecbeta_1^{\intercal}\matsig_{11}\inv\vecbeta_1+\vecbeta_{2\mid1}^{\intercal}\matsig_{2\mid1}\vecbeta_{2\mid1}}\right)^\frac{\lambda-\frac{d_1}{2}-\frac{d_2}{2}}{2}\left [\frac{\omega+\vecbeta_1^{\intercal}\matsig_{11}\inv\vecbeta_1}{\omega +\delta(\vecy_1,\vecmu_1 \mid \matsig_{11} )}\right]^{\frac{\lambda-d_{1}/2}{2}}}{(2\pi)^{\frac{d_2}{2}}|\matsig_{22}-\matsig_{12}^{\intercal}\matsig_{11}\inv\matsig_{12}|^{\frac{1}{2}}}\\
&\times\frac{K_{\lambda-\frac{d_1}{2}-\frac{d_2}{2}}\left(\sqrt{(\omega +\delta(\vecy_1,\vecmu_1 \mid \matsig_{11} )+\delta(\vecy_2, \vecmu_{2\mid1}\mid \matsig_{2\mid1}))(\omega+\vecbeta_1^{\intercal}\matsig_{11}\inv\vecbeta_1+\vecbeta_{2\mid1}^{\intercal}\matsig_{2\mid1}\vecbeta_{2\mid1})} \right)}{K_{\lambda-\frac{d_{1}}{2}}\left(\sqrt{(\omega +\delta(\vecy_1,\vecmu_1 \mid \matsig_{11}))(\omega+\vecbeta_1^{\intercal}\matsig_{11}\inv\vecbeta_1)}\right)\text{exp}(-(\vecy_2-\vecmu_{2\mid1})^{\intercal}\matsig_{2\mid1}\inv\vecbeta_{2\mid1})}.
\end{align*}
Set $\lambda_{2\mid1}=\lambda-\frac{d_1}{2}$, $\chi_{2\mid1}=\omega+\delta(\vecy_1, \vecmu_1\mid \matsig_{11})$, and $\psi_{2\mid1}=\omega+\vecbeta_1^{\intercal}\matsig_{11}\inv\vecbeta_1$, then we obtain
\begin{align*}
f_{\vecY_2 \mid \vecY_1}(\vecy_2 \mid \vecy_1) &=\left[\frac{\chi_{2\mid1}+\delta(\vecy_2, \vecmu_{2\mid1}\mid \matsig_{2\mid1})}{\psi_{2\mid1}+\vecbeta_{2\mid1}^{\intercal}\matsig_{2\mid1}\vecbeta_{2\mid1}}\right]^{\frac{\lambda_{2\mid1}-\frac{d_2}{2}}{2}}\\
&\times\frac{\left(\frac{\psi_{2\mid1}}{\chi_{2\mid1}}\right)^{\frac{\lambda_{2\mid1}}{2}}K_{\lambda_{2\mid1}-\frac{d_2}{2}}\left(\sqrt{(\psi_{2\mid1}+\vecbeta_{2\mid1}^{\intercal}\matsig_{2\mid1}\vecbeta_{2\mid1})(\chi_{2\mid1}+\delta(\vecy_2, \vecmu_{2\mid1}\mid \matsig_{2\mid1}))}\right)} {(2\pi)^{\frac{d_2}{2}}|\matsig_{2\mid1}|^{\frac{1}{2}}K_{\lambda_{2\mid1}}(\sqrt{\chi_{2\mid1}\psi_{2\mid1}})\text{exp}(-(\vecy_2-\vecmu_{2\mid1})^{\intercal}\matsig_{2\mid1}\inv\vecbeta_{2\mid1})}.
\end{align*}

Comparison with \eqref{eqn:ghd} reveals that this is a generalized hyperbolic distribution in the parameterization of \citet{mcneil05} with
\begin{align*}
\lambda_{2\mid1}&= \lambda-\frac{d_1}{2},&\chi_{2\mid1}&=\omega+(\vecy_1-\vecmu_1)^{\intercal}\matsig_{11}\inv(\vecy_1-\vecmu_1),\\
\psi_{2\mid1}&=\omega+\vecbeta_1^{\intercal}\matsig_{11}^{\intercal}\vecbeta,&\vecmu_{2\mid1}&=\vecmu_2+\matsig_{12}^{\intercal}\matsig_{11}\inv(\vecy_1-\vecmu_1),\\
\matsig_{2\mid1}&=\matsig_{22}-\matsig_{12}^{\intercal}\matsig_{11}\inv\matsig_{12},&\vecbeta_{2\mid1}&=\vecbeta_2-\matsig_{12}^{\intercal}\matsig_{11}\inv\vecbeta_1.
\end{align*}

\section{MST with Incomplete Data}\label{sec:MST}
Analogous to the MGHD model (\ref{eqn:mghd}), the MST model takes the density
\begin{equation}
\label{eqn:MST}
f_{\text{MST}}(\vecY_i\mid\matTheta) = \gsum \pi_g f_{\text{ST}} (\vecY_i\mid v_g,\vecmu_g,\matsig_g,\vecbeta_g),
\end{equation}
where $\matTheta=(\mathbf{\pi},\textbf{v}_g, \vecmu_g, \matsig_g, \vecbeta_g)$ with $\textbf{v}_g=(v_1,\ldots,v_g)$ and $\pi_g,\vecmu_g,\matsig_g$, and $\vecbeta_g$ are as defined above. By introducing the group membership variables $\vecZ_i\sim\mathcal{M}(1;\pi_1,\ldots,\pi_G)$, convenient three-layer hierarchical representations are given by
\begin{align}
\nonumber
\vecY_i\mid w_{ig},z_{ig}=1 &\sim \mathcal{N}(\vecmu_g+w_{ig}\vecbeta_g,w_{ig}\matsig_g)\\
\label{eqn:MSTthreelevel}
W_{ig}\mid z_{ig}=1&\sim\text{IG}(v_g/2,v_g/2).\\
\nonumber
\vecZ_i&\sim\mathcal{M}(1;\pi_1,\ldots,\pi_G)
\end{align}
Assume that the matrix $\vecY=(\vecY^{\text{o}\intercal},\vecY^{\text{m}\intercal})^{\intercal}$ contains missing data. For each $\vecy_i=(\vecy_i^{\text{o}\intercal},\vecy_i^{\text{m}\intercal})^{\intercal}$, we write $\vecmu_g=(\vecmu_{g,i}^{\text{o}\intercal},\vecmu_{g,i}^{\text{m}\intercal})^{\intercal}$, $\vecbeta_g=(\vecbeta_{g,i}^{\text{o}\intercal},\vecbeta_{g,i}^{\text{m}\intercal})^{\intercal}$, and finally the $g$th dispersion matrix $\matsig_g$ is partitioned as in (\ref{eqn:matsigmiss}). Hence, based on (\ref{eqn:MSTthreelevel}), we have the following conditional distributions:
\begin{itemize}
\item The marginal distribution of $\vecY_i^{\text{o}}$ is 
	\begin{equation}
	\nonumber
	\vecY_i^{\text{o}} \sim \gsum \pi_g f_{\text{ST},p_i^{\text{o}}}(\lambda_g,\omega_g,\vecmu_{g,i}^{\text{o}}, \matsig_{g,i}^{\text{oo}}, \vecbeta_{g,i}^{\text{o}}),
	\end{equation}
	where $p_i^{\text{o}}$ is the dimension corresponding to the observed component $\vecy_i^{\text{o}}$, which should be exactly written as $p_i^{\text{o}_i}$ but here is simplified.
\item The conditional distribution of $\vecY_i^{\text{m}}$ given $\vecy_i^{\text{o}}$ and $z_{ig}=1$, according to Proposition~6, is 
	\begin{equation}
	\label{eqn:ymissconpdfskewt}
	\vecY_i^{\text{m}} \mid \vecy_i^{\text{o}}, z_{ig}=1\sim \text{GH}_{p-p_i^{\text{o}}}(\lambda_{g,i}^{\text{m}\mid \text{o}},\chi_{g,i}^{\text{m}\mid \text{o}},\psi_{g,i}^{\text{m}\mid \text{o}},\vecmu_{g,i}^{\text{m}\mid \text{o}},\matsig_{g,i}^{\text{m}\mid \text{o}},\vecbeta_{g,i}^{\text{m}\mid \text{o}}),  
	\end{equation}
	where
	\begin{align}
	\nonumber
	\lambda_{g,i}^{\text{m}\mid \text{o}}&= -\frac{v_g+p_i^{\text{o}}}{2},&\psi_{g,i}^{\text{m}\mid \text{o}}&=v_g+(\vecy_i^{\text{o}}-\vecmu_{g,i}^{\text{o}})^{\intercal}(\matsig_{g,i}^{\text{oo}})\inv(\vecy_i^{\text{o}}-\vecmu_{g,i}^{\text{o}}),\\
	\nonumber
	\psi_{g,i}^{\text{m}\mid \text{o}}&=\vecbeta_{g,i}^{\text{o}\intercal}(\matsig_{g,i}^{\text{oo}})\inv\vecbeta_{g,i}^{\text{o}},&\vecmu_{g,i}^{\text{m}\mid \text{o}}&=\vecmu_{g,i}^{\text{m}}+(\matsig_{g,i}^{\text{om}})^{\intercal}(\matsig_{g,i}^{\text{oo}})\inv(\vecy_i^{\text{o}}-\vecmu_{g,i}^{\text{o}}),\\
	\nonumber
	\matsig_{g,i}^{\text{m}\mid \text{o}}&=\matsig_{g,i}^{\text{mm}}-(\matsig_{g,i}^{\text{om}})^{\intercal}(\matsig_{g,i}^{\text{oo}})\inv\matsig_{g,i}^{\text{om}},&\vecbeta_{g,i}^{\text{m}\mid \text{o}}&=\vecbeta_{g,i}^{\text{m}}-(\matsig_{g,i}^{\text{om}})^{\intercal}(\matsig_{g,i}^{\text{oo}})\inv\vecbeta_{g,i}^{\text{o}}.
	\end{align}
\item The conditional distribution of $\vecY_i^{\text{m}}$ given $\vecy_i^{\text{o}}, w_{ig}$, and $z_{ig}=1$ is 
        \begin{equation}
        \vecY_i^{\text{m}} \mid \vecy_i^{\text{o}}, w_{ig}, z_{ig}=1\sim \mathcal{N}_{p-p_i^{\text{o}}}(\vecmu_{g,i}^{\text{m}\mid \text{o}}+w_{ig}\vecbeta_{g,i}^{\text{m}\mid \text{o}}, w_{ig}\matsig_{g,i}^{\text{m}\mid \text{o}}).
        \end{equation}
\item The conditional distribution of $W_i$ given $\vecy_i^{\text{o}}$ and $z_{ig}=1$ is
	\begin{equation}
	\label{eqn:wconpdfskewt}
	W_{ig} \mid \vecy_i^{\text{o}},z_{ig}=1 \sim \text{GIG}\left(\vecbeta_{g,i}^{\text{o}\intercal}(\matsig_{g,i}^{\text{oo}})\inv\vecbeta_{g,i}^{\text{o}}, v_g+\delta(\vecy_i^{\text{o}},\vecmu_{g,i}^{\text{o}}\mid \matsig_{g,i}^{\text{oo}}), -\frac{v_g+p_i^{\text{o}}}{2}\right).
	\end{equation} 
\end{itemize}

As in the case of the MGHD model with incomplete data, the complete data consists of the observed $\vecy_i$, the missing group membership $z_{ig}$, the latent $w_{ig}$, as well as the actual missing data $\vecy_i^{\text{m}}$, for $i=1,\ldots,n$ and $g=1,\ldots,G$. Again, the complete data log-likelihood function is given by
\begin{align}
\label{eqn:MSTloglikmiss}
l_{\text{c}}(\matTheta) = \isum\gsum z_{ig}\left[\log\pi_g+\log\phi(\vecy_i^{\text{o}},\vecy_i^{\text{m}}\mid\vecmu_g+w_{ig}\vecbeta_g,w_{ig}\matsig_g)+\log f_{\text{IG}}(w_{ig}\mid v_g/2,v_g/2)\right].
\end{align}
Furthermore, one can simplify (\ref{eqn:MSTloglikmiss}) to
\begin{equation}
\label{eqn:MSTloglik}
\begin{split}
l_{\text{c}}(\matTheta)&
=\isum\gsum z_{ig}\log\pi_g+\isum\gsum z_{ig}\left[ -\frac{p}{2}\log(2\pi)-\frac{p}{2}\log w_{ig}+\frac{1}{2}\log |\matsig_g\inv|\right]\\
&-\frac{1}{2}\isum\gsum\text{tr}\left( \matsig_g\inv z_{ig} \frac{1}{w_{ig}}\begin{pmatrix} (\vecy_i^{\text{o}}-\vecmu_{g,i}^{\text{o}})(\vecy_i^{\text{o}}-\vecmu_{g,i}^{\text{o}})^{\intercal} &(\vecy_i^{\text{o}}-\vecmu_{g,i}^{\text{o}})(\vecy_i^{\text{m}}-\vecmu_{g,i}^{\text{m}})^{\intercal}\\(\vecy_i^{\text{m}}-\vecmu_{g,i}^{\text{m}})^{\intercal}(\vecy_i^{\text{o}}-\vecmu_{g,i}^{\text{o}})&(\vecy_i^{\text{m}}-\vecmu_{g,i}^{\text{m}})(\vecy_i^{\text{m}}-\vecmu_{g,i}^{\text{m}})^{\intercal}\end{pmatrix} \right)\\
&+\frac{1}{2}\isum\gsum\text{tr}\left(  \matsig_g\inv z_{ig} \begin{pmatrix}\vecbeta_{g,i}^{\text{o}}\\\vecbeta_{g,i}^{\text{m}}\end{pmatrix}\begin{pmatrix} (\vecy_i^{\text{o}}-\vecmu_{g,i}^{\text{o}})^{\intercal}&(\vecy_i^{\text{m}}-\vecmu_{g,i}^{\text{m}})^{\intercal}\end{pmatrix}\right)\\
&+\frac{1}{2}\isum\gsum\text{tr}\left(  \matsig_g\inv z_{ig} \begin{pmatrix} \vecy_i^{\text{o}}-\vecmu_{g,i}^{\text{o}}\\\vecy_i^{\text{m}}-\vecmu_{g,i}^{\text{m}}\end{pmatrix}\begin{pmatrix}\vecbeta_{g,i}^{\text{o}\intercal}&\vecbeta_{g,i}^{\text{m}\intercal}\end{pmatrix}\right)-\frac{1}{2}\isum\gsum z_{ig}w_{ig}\vecbeta_{g,i}^{\intercal}\matsig_g\inv\vecbeta_{g,i}\\
&+\isum\gsum z_{ig}\left[\frac{v_g}{2}\log\left(\frac{v_g}{2}\right)-\log\Gamma\left(\frac{v_g}{2}\right)-\left(\frac{v_g}{2}+1\right)\log w_{ig}-\frac{v_g}{2w_{ig}}\right].
\end{split}
\end{equation}
 
On the $k$th iteration of the E-step, the expected value of the complete-data log-likelihood is computed given the observed data $\vecY^{\text{o}}$ and the current parameter updates $\matTheta^{(k)}$. Denote by $\tau_{ig}^{(k)}$ the \textit{a~posteriori} probability that the $i$th observation belongs to the $g$th component of the mixture. Specifically, it can be calculated as
\begin{equation}
\nonumber
\tau_{ig}^{(k+1)} \colonequals \E(Z_{ig}\mid \vecy_i^{\text{o}},\matTheta^{(k)})=\frac{\pi_g^{(k)}f_{\text{ST},p_i^{\text{o}}}(\vecy_i^{\text{o}};v_g^{(k)},\vecmu_{g,i}^{\text{o}(k)},\matsig_{g,i}^{\text{oo}(k)},\vecbeta_{g,i}^{\text{o}(k)})}{\sum_{l=1}^{G}\pi_l^{(k)}f_{\text{ST},p_i^{\text{o}}}(\vecy_i^{\text{o}};v_l^{(k)},\vecmu_{l,i}^{\text{o}(k)},\matsig_{l,i}^{\text{oo}(k)},\vecbeta_{l,i}^{\text{o}(k)})}.
\end{equation}

Given the observed data $\vecy^{\text{o}}$, the current parameter updates $\matTheta^{(k)}$, and conditional distributions (\ref{eqn:ymissconpdfskewt}) and (\ref{eqn:wconpdfskewt}), taking expectations for (\ref{eqn:MSTloglik}) leads to the following expectation updates in the E-step:
\begin{align*}
A_{ig}^{(k)}&\colonequals \E(W_{ig}\mid \vecy_i^{\text{o}},z_{ig}=1;\matTheta^{(k)}) = \sqrt{\frac{v_g^{(k)}+\delta(\vecy_i^{\text{o}},\vecmu_{g,i}^{\text{o}(k)}\mid\matsig_{g,i}^{\text{oo}(k)})}{\vecbeta_{g,i}^{o(k)\intercal}(\matsig_{g,i}^{\text{oo}(k)})\inv\vecbeta_{g,i}^{\text{o}(k)}}}\\
&\qquad\times\frac{K_{-(v_g^{(k)}+p_i^0)/2+1}\left(\sqrt{(v_g^{(k)}+\delta(\vecy_i^{\text{o}},\vecmu_{g,i}^{\text{o}(k)}\mid\matsig_{g,i}^{\text{oo}(k)}))(\vecbeta_{g,i}^{o(k)\intercal}(\matsig_{g,i}^{\text{oo}(k)})\inv\vecbeta_{g,i}^{\text{o}(k)})}\right)}{K_{-(v_g^{(k)}+p_i^0)/2}\left(\sqrt{(v_g^{(k)}+\delta(\vecy_i^{\text{o}},\vecmu_{g,i}^{\text{o}(k)}\mid\matsig_{g,i}^{\text{oo}(k)}))(\vecbeta_{g,i}^{o(k)\intercal}(\matsig_{g,i}^{\text{oo}(k)})\inv\vecbeta_{g,i}^{\text{o}(k)})}\right)},\\
B_{ig}^{(k)}&\colonequals\E({1}/{W_{ig}}\mid \vecy_i^{\text{o}},z_{ig}=1;\matTheta^{(k)}) \\
&= \frac{v_g^{(k)}+p_i^{\text{o}}}{v_g^{(k)}+\delta(\vecy_i^{\text{o}},\vecmu_{g,i}^{\text{o}(k)}\mid\matsig_{g,i}^{\text{oo}(k)})}+\sqrt{\frac{\vecbeta_{g,i}^{o(k)\intercal}(\matsig_{g,i}^{\text{oo}(k)})\inv\vecbeta_{g,i}^{\text{o}(k)}}{v_g^{(k)}+\delta(\vecy_i^{\text{o}},\vecmu_{g,i}^{\text{o}(k)}\mid\matsig_{g,i}^{\text{oo}(k)})}}\\
&\qquad\times\frac{K_{-(v_g^{(k)}+p_i^0)/2+1}\left(\sqrt{(v_g^{(k)}+\delta(\vecy_i^{\text{o}},\vecmu_{g,i}^{\text{o}(k)}\mid\matsig_{g,i}^{\text{oo}(k)}))(\vecbeta_{g,i}^{o(k)\intercal}(\matsig_{g,i}^{\text{oo}(k)})\inv\vecbeta_{g,i}^{\text{o}(k)})}\right)}{K_{-(v_g^{(k)}+p_i^0)/2}\left(\sqrt{(v_g^{(k)}+\delta(\vecy_i^{\text{o}},\vecmu_{g,i}^{\text{o}(k)}\mid\matsig_{g,i}^{\text{oo}(k)}))(\vecbeta_{g,i}^{o(k)\intercal}(\matsig_{g,i}^{\text{oo}(k)})\inv\vecbeta_{g,i}^{\text{o}(k)})}\right)},\\
C_{ig}^{(k)}&\colonequals\E(\log W_{ig}\mid \vecy_i^{\text{o}},z_{ig}=1;\matTheta^{(k)})=\log\left( \sqrt{\frac{v_g^{(k)}+\delta(\vecy_i^{\text{o}},\vecmu_{g,i}^{\text{o}(k)}\mid\matsig_{g,i}^{\text{oo}(k)})}{\vecbeta_{g,i}^{o(k)\intercal}(\matsig_{g,i}^{\text{oo}(k)})\inv\vecbeta_{g,i}^{\text{o}(k)}}} \right)\\
&+\left.\frac{\partial}{\partial t}\log\left\{K_{t}\left(\sqrt{(v_g^{(k)}+\delta(\vecy_i^{\text{o}},\vecmu_{g,i}^{\text{o}(k)}\mid\matsig_{g,i}^{\text{oo}(k)}))(\vecbeta_{g,i}^{o(k)\intercal}(\matsig_{g,i}^{\text{oo}(k)})\inv\vecbeta_{g,i}^{\text{o}(k)})}\right)\right\}\right\rvert_{t=-(v_g^{(k)}+p_i^{\text{o}})/2},\\
\hat{\vecy}_{ig}^{\text{m}(k)}&\colonequals\E(\vecY_i^{\text{m}}\mid\vecy_i^{\text{o}},z_{ig}=1)=\vecmu_{g,i}^{\text{m}\mid \text{o}(k)}+A_{ig}^{(k)}\vecbeta_{g,i}^{\text{m}\mid \text{o}(k)},\\
\tilde{\vecy}_{ig}^{\text{m}(k)}&\colonequals\E(({1}/{W_{i}})\vecY_i^{\text{m}}\mid\vecy_i^{\text{o}},z_{ig}=1)=B_{ig}^{(k)}\vecmu_{g,i}^{\text{m}\mid \text{o}(k)}+\vecbeta_{g,i}^{\text{m}\mid \text{o}(k)},\\
\tilde{\tilde{\vecy}}_{ig}^{\text{m}(k)}&\colonequals\E(({1}/{w_{i}})\vecY_i^{\text{m}}\vecY_i^{m\intercal}
\mid\vecy_i^{\text{o}},z_{ig}=1)=\matsig_{g,i}^{\text{m}\mid \text{o}(k)}+B_{ig}^{(k)}\vecmu_{g,i}^{\text{m}\mid \text{o}(k)}(\vecmu_{g,i}^{\text{m}\mid \text{o}(k)})^{\intercal}\\
&\qquad\qquad\qquad\qquad+\vecmu_{g,i}^{\text{m}\mid \text{o}(k)}(\vecbeta_{g,i}^{\text{m}\mid \text{o}(k)})^{\intercal}+\vecbeta_{g,i}^{\text{m}\mid \text{o}(k)}(\vecmu_{g,i}^{\text{m}\mid \text{o}(k)})^{\intercal}+A_{ig}^{(k)}\vecbeta_{g,i}^{\text{m}\mid \text{o}(k)}(\vecbeta_{g,i}^{\text{m}\mid \text{o}(k)})^{\intercal}.
\end{align*} 
 
For convenience, let $n_g^{(k)} = \isum \tau_{ig}^{(k)}$, $\bar{A}_g^{(k)} = 1/n_g^{(k)}\isum \tau_{ig}^{(k)}A_{ig}^{(k)}$, $\bar{B}_g^{(k)} = 1/n_g^{(k)}\isum \tau_{ig}^{(k)}B_{ig}^{(k)}$, and $\bar{C}_g^{(k)} = 1/n_g^{(k)}\isum \tau_{ig}^{(k)}C_{ig}^{(k)}$. On the $k$th iteration of the M-step, we get updates for the parameter estimates of the mixture as follows:
 \begin{align*}
\pi_g^{(k+1)} &= \frac{n_g^{(k)}}{n},\\
{\vecmu}_g^{(k+1)} &= \frac{1}{\isum\hat{\tau}_{ig}^{(k)}(\bar{A}_g^{(k)}B_{ig}^{(k)}-1)}\isum \hat{\tau}_{ig}^{(k)}\begin{pmatrix}(\bar{A}_g^{(k)}B_{ig}^{(k)}-1)\vecy_i^{\text{o}}\\\bar{A}_g^{(k)}\tilde{\vecy}_{ig}^{\text{m}(k)}-\hat{\vecy}_{ig}^{\text{m}(k)}\end{pmatrix},\\
{\vecbeta}_g^{(k+1)}&=\frac{1}{\isum\hat{\tau}_{ig}^{(k)}(\bar{A}_g^{(k)}B_{ig}^{(k)}-1)}\isum \hat{\tau}_{ig}^{(k)}\begin{pmatrix}(\bar{B}_g^{(k)}-B_{ig}^{(k)})\vecy_i^{\text{o}}\\\bar{B}_g^{(k)}\hat{\vecy}_{ig}^{\text{m}(k)}-\tilde{\vecy}_{ig}^{\text{m}(k)}\end{pmatrix},\\
{\matsig}_g^{(k+1)}&=\frac{1}{n_g^{(k)}}\isum\hat{\tau}_{ig}^{(k)}{\matsig}_{ig}^{(k+1)}-(\bar{\vecy}_g-{\vecmu}_g^{(k+1)}){\vecbeta}_g^{(k+1)\intercal}-{\vecbeta}_g^{(k+1)}(\bar{\vecy}_g-{\vecmu}_g^{(k+1)})^{\intercal}+\bar{A}_g^{(k+1)}{\vecbeta}_g^{(k+1)}{\vecbeta}_g^{(k+1)\intercal},
\end{align*}
where 
\begin{align*}
\bar{\vecy}_g&=\frac{1}{n_g^{(k+1)}}\isum\hat{\tau}_{ig}^{(k+1)}\begin{pmatrix}\vecy_i^{\text{o}}\\\hat{\vecy}_{ig}^{\text{m}(k+1)}\end{pmatrix},\\
{\matsig}_{ig}^{(k+1)} &= \begin{pmatrix}B_{ig}^{(k+1)}(\vecy_i^{\text{o}}-\vecmu_{g,i}^{\text{o}(k+1)})(\vecy_i^{\text{o}}-\vecmu_{g,i}^{\text{o}(k+1)})^{\intercal}&(\vecy_i^{\text{o}}-\hat{\vecmu}_g^{\text{o}(k+1)})(\tilde{\vecy}_{ig}^{\text{m}(k+1)}-B_{ig}^{(k+1)}\hat{\vecmu}_g^{\text{m}(k+1)})^{\intercal}\\(\tilde{\vecy}_{ig}^{\text{m}(k+1)}-B_{ig}^{(k+1)}\hat{\vecmu}_g^{\text{m}(k+1)})(\vecy_i^{\text{o}}-\vecmu_{g,i}^{\text{o}(k+1)})^{\intercal}&\mathbf{k}_{ig}^{\text{m}(k+1)}\end{pmatrix},
\end{align*}
where $$\mathbf{k}_{ig}^{\text{m}(k+1)}=\tilde{\tilde{\vecy}}_{ig}^{\text{m}(k+1)}-\tilde{\vecy}_{ig}^{\text{m}(k)}\hat{\vecmu}_g^{\text{m}(k+1)T}-\hat{\vecmu}_g^{\text{m}(k+1)}\tilde{\vecy}_i^{\text{m}(k)\intercal}+B_{ig}^{(k)}\hat{\vecmu}_g^{\text{m}(k+1)}\hat{\vecmu}_g^{\text{m}(k+1)\intercal}.$$

Finally, as for the degree of freedom parameter $v_g$, the update does not exist in closed form. The update $v_g^{(k+1)}$ is the solution of
\begin{equation}
\log\left(\frac{v_g^{(k+1)}}{2}\right) +1-\varphi\left(\frac{v_g^{(k+1)}}{2}\right)-\frac{1}{n_g^{(k)}}\isum\tau_{ig}(C_{ig}^{(k)}+B_{ig}^{(k)})=0,
\end{equation}
where $\varphi(\cdot)$ is the digamma function.

\section{Results from Simulation Studies}\label{app:tables}
The results from the simulation studies are summarized in Tables~\ref{sim1parresult}, ~\ref{sim3parresult} and~\ref{table:SimNG}.

\begin{sidewaystable}[ht]
\caption{Key model parameters as well as means, standard deviations and bias of the associated parameter estimations from the 100 runs for the first simulation experiment.}
\label{sim1parresult}
	\centering
	\tiny{
	\begin{tabular*}{1.1\textwidth}{@{\extracolsep{\fill}}cccccccccc}
		\hline
		\multicolumn{10}{c}{Sim1 using MGHD}\\
		\hline
		\multicolumn{10}{c}{$r=0.05$}\\
		&\multicolumn{3}{c}{Pattern 1}&\multicolumn{3}{c}{Pattern 2}&\multicolumn{3}{c}{MCAR}\\
\cline{2-4}\cline{5-7}\cline{8-10}
	&
	Mean&Std. dev&Bias&Mean&Std. dev.&Bias&Mean&Std. dev.&Bias\\
	\hline
	$\vecmu_1$
	&$(0.70,-3.30)'$
	&$(1.06,0.94)'$
	&$(-0.30,-0.30)'$
	&$(0.45,-3.40)'$
	&$(1.23,1.08)'$
	&$(-0.55,-0.40)'$
	&$(0.64,-3.22)'$
	&$(0.86,0.79)'$
	&$(-0.36,-0.22)'$
	\\
	$\vecmu_2$
	&$(-0.74,3.47)'$
	&$(2.11,2.49)'$
	&$(0.26,0.47)'$
	&$(-0.57,3.41)'$
	&$(2.62,2.17)'$
	&$(0.43,0.41)'$
	&$(-0.64,3.42)'$
	&$(2.54,2.20)'$
	&$(0.36,0.42)'$
	\\
	$\vecbeta_1$
	&$(1.59,1.59)'$
	&$(1.30,1.16)'$
	&$(0.59,0.59)'$
	&$(1.91,1.74)'$
	&$(1.49,1.35)'$
	&$(0.91,0.74)'$
	&$(1.67,1.50)'$
	&$(1.06,0.98)'$
	&$(0.67,0.50)'$
	\\
	
	$\vecbeta_2$
	&$(-1.73,-1.98)'$
	&$(2.50,2.92)'$
	&$(-0.73,-0.98)'$
	&$(-1.96,-1.91)'$
	&$(3.10,2.56)'$
	&$(-0.96,-0.91)'$
	&$(-1.86,-1.94)'$
	&$(2.32,2.09)'$
	&$(-0.86,-0.94)'$
	\\
	$\vecmu_1+\vecbeta_1$
	&$(2.29,-1.71)'$
	&$(0.26,0.25)'$
	&$(0.29,0.29)'$
	&$(2.36,-1.66)'$
	&$(0.32,0.32)'$
	&$(0.36,0.34)'$
	&$(2.31,-1.71)'$
	&$(0.25,0.26)'$
	&$(0.31,0.29)'$
	\\
	$\vecmu_2+\vecbeta_2$
	&$(-2.47,1.48)'$
	&$(0.44,0.47)'$
	&$(-0.47,-0.52)'$
	&$(-2.54,1.50)'$
	&$(0.51,0.43)'$
	&$(-0.54,-0.50)'$
	&$(-2.50,1.48)'$
	&$(0.50,0.54)'$
	&$(-0.50,-0.52)'$
	\\
	$\matsig_1$
		&$\left[\begin{array}{cc}1.88&1.51\\1.51&1.90\end{array}\right]$
		&$\left[\begin{array}{cc}0.32&0.27\\0.27&0.30\end{array}\right]$
		&$\left[\begin{array}{cc}0.21&0.18\\0.18&0.23\end{array}\right]$
		&$\left[\begin{array}{cc}1.96&1.57\\1.57&1.98\end{array}\right]$
		&$\left[\begin{array}{cc}0.34&0.28\\0.28&0.33\end{array}\right]$
		&$\left[\begin{array}{cc}0.29&0.24\\0.24&0.31\end{array}\right]$
		&$\left[\begin{array}{cc}1.95&1.57\\1.57&1.97\end{array}\right]$			
		&$\left[\begin{array}{cc}0.36&0.30\\0.30&0.34\end{array}\right]$
		&$\left[\begin{array}{cc}0.28&0.25\\0.25&0.30\end{array}\right]$
		\\
	$\matsig_2$
		&$\left[\begin{array}{cc}4.42&3.55\\3.55&4.48\end{array}\right]$
		&$\left[\begin{array}{cc}0.66&0.58\\0.58&0.68\end{array}\right]$
		&$\left[\begin{array}{cc}1.09&0.88\\0.88&1.15\end{array}\right]$
		&$\left[\begin{array}{cc}4.38&3.53\\3.53&4.43\end{array}\right]$
		&$\left[\begin{array}{cc}0.76&0.66\\0.66&0.76\end{array}\right]$
		&$\left[\begin{array}{cc}1.05&0.86\\0.86&1.10\end{array}\right]$
		&$\left[\begin{array}{cc}4.43&3.56\\3.56&4.50\end{array}\right]$
		&$\left[\begin{array}{cc}0.68&0.58\\0.58&0.68\end{array}\right]$
		&$\left[\begin{array}{cc}1.10&0.89\\0.89&1.17\end{array}\right]$
		\\
	$\lambda_1$
	&$-2.26$
	&$1.27$
	&$-1.76$
	&$-2.70$
	&$1.70$
	&$-2.20$
	&$-2.51$
	&$1.26$
	&$-2.01$
	\\
	$\lambda_2$
	&$2.93$
	&$1.40$
	&$1.93$
	&$2.79$
	&$1.40$
	&$1.79$
	&$2.88$
	&$1.40$
	&$1.88$
	\\
	$\pi_1$&$0.50$
	&$0.00$
	&$0.00$
	&$0.50$
	&$0.01$
	&$0.00$
	&$0.50$
	&$0.01$
	&$0.00$
	\\
	$\pi_2$&$0.50$
	&$0.00$
	&$0.00$
	&$0.50$
	&$0.01$
	&$0.00$
	&$0.50$
	&$0.01$
	&$0.00$
	\\
	ARI&$0.96$
	&$0.02$
	&
	&$0.96$
	&$0.02$
	&
	&$0.95$
	&$0.09$
	&
	\\
	\hline	
		\multicolumn{10}{c}{$r=0.15$}\\
		&\multicolumn{3}{c}{Pattern 1}&\multicolumn{3}{c}{Pattern 2}&\multicolumn{3}{c}{MCAR}\\
\cline{2-4}\cline{5-7}\cline{8-10}
	&
	Mean&Std. dev&Bias&Mean&Std. dev.&Bias&Mean&Std. dev.&Bias\\
	\hline
	$\vecmu_1$
	&$(0.81,-3.14)'$
	&$(1.14,1.14)'$
	&$(-0.19,-0.14)'$
	&$(0.72,-3.24)'$
	&$(1.10,0.89)'$
	&$(-0.28,-0.24)'$
	&$(0.75,-3.23)'$
	&$(1.24,1.11)'$
	&$(-0.25,-0.23)'$
	\\
	$\vecmu_2$
	&$(-0.69,3.45)'$
	&$(2.58,2.25)'$
	&$(0.31,0.45)'$
	&$(-0.49,3.47)'$
	&$(2.90,3.00)'$
	&$(0.51,0.47)'$
	&$(-0.67,3.34)'$
	&$(1.95,1.78)'$
	&$(0.33,0.34)'$
	\\
	$\vecbeta_1$
	&$(1.46,1.41)'$
	&$(1.40,1.41)'$
	&$(0.46,0.41)'$
	&$(1.38,1.11)'$
	&$(1.49,1.35)'$
	&$(0.38,0.11)'$
	&$(1.59,1.50)'$
	&$(1.59,1.50)'$
	&$(0.59,0.50)'$
	\\
	
	$\vecbeta_2$
	&$(-1.77,-1.95)'$
	&$(3.01,2.61)'$
	&$(-0.77,-0.95)'$
	&$(-2.06,-1.99)'$
	&$(3.40,3.50)'$
	&$(-1.06,-0.99)'$
	&$(-1.72,-1.80)'$
	&$(3.00,2.62)'$
	&$(-0.72,-0.80)'$
	\\
	$\vecmu_1+\vecbeta_1$
	&$(2.28,-1.72)'$
	&$(0.31,0.32)'$
	&$(0.28,0.28)'$
	&$(2.29,-1.69)'$
	&$(0.31,0.27)'$
	&$(0.29,0.31)'$
	&$(2.33,-1.61)'$
	&$(0.45,0.39)'$
	&$(0.33,0.39)'$
	\\
	$\vecmu_2+\vecbeta_2$
	&$(-2.47,1.51)'$
	&$(0.48,0.40)'$
	&$(-0.47,-0.49)'$
	&$(-2.55,1.48)'$
	&$(0.53,0.58)'$
	&$(-0.55,-0.52)'$
	&$(-2.39,1.54)'$
	&$(0.45,0.50)'$
	&$(-0.39,-0.48)'$
	\\
	$\matsig_1$
		&$\left[\begin{array}{cc}1.94&1.55\\1.55&1.95\end{array}\right]$
		&$\left[\begin{array}{cc}0.38&0.34\\0.34&0.38\end{array}\right]$
		&$\left[\begin{array}{cc}0.27&0.22\\0.22&0.28\end{array}\right]$
		&$\left[\begin{array}{cc}1.90&1.51\\1.51&1.90\end{array}\right]$
		&$\left[\begin{array}{cc}0.36&0.30\\0.30&0.33\end{array}\right]$
		&$\left[\begin{array}{cc}0.23&0.18\\0.18&0.23\end{array}\right]$
		&$\left[\begin{array}{cc}1.91&1.51\\1.51&1.93\end{array}\right]$			
		&$\left[\begin{array}{cc}0.44&0.33\\0.33&0.43\end{array}\right]$
		&$\left[\begin{array}{cc}0.24&0.18\\0.18&0.26\end{array}\right]$
		\\
	$\matsig_2$
		&$\left[\begin{array}{cc}4.28&3.41\\3.41&4.28\end{array}\right]$
		&$\left[\begin{array}{cc}0.72&0.62\\0.62&0.69\end{array}\right]$
		&$\left[\begin{array}{cc}0.95&0.74\\0.74&0.95\end{array}\right]$
		&$\left[\begin{array}{cc}4.30&3.43\\3.43&4.33\end{array}\right]$
		&$\left[\begin{array}{cc}0.67&0.60\\0.60&0.70\end{array}\right]$
		&$\left[\begin{array}{cc}0.74&0.76\\0.76&0.77\end{array}\right]$
		&$\left[\begin{array}{cc}4.42&3.55\\3.55&4.50\end{array}\right]$
		&$\left[\begin{array}{cc}0.73&0.68\\0.68&0.75\end{array}\right]$
		&$\left[\begin{array}{cc}1.09&0.88\\0.88&1.17\end{array}\right]$
		\\
	$\lambda_1$
	&$-2.51$
	&$1.38$
	&$-2.01$
	&$-2.47$
	&$1.67$
	&$-1.98$
	&$-3.14$
	&$1.77$
	&$-2.64$
	\\
	$\lambda_2$
	&$3.12$
	&$1.76$
	&$2.12$
	&$3.24$
	&$1.43$
	&$2.24$
	&$2.52$
	&$1.44$
	&$1.52$
	\\
	$\pi_1$&$0.50$
	&$0.01$
	&$0.00$
	&$0.50$
	&$0.01$
	&$0.00$
	&$0.50$
	&$0.01$
	&$0.00$
	\\
	$\pi_2$&$0.50$
	&$0.01$
	&$0.00$
	&$0.50$
	&$0.01$
	&$0.00$
	&$0.50$
	&$0.01$
	&$0.00$
	\\
	ARI&$0.93$
	&$0.02$
	&
	&$0.93$
	&$0.02$
	&
	&$0.93$
	&$0.09$
	&
	\\
	\hline		
			\multicolumn{10}{c}{$r=0.30$}\\
		&\multicolumn{3}{c}{Pattern 1}&\multicolumn{3}{c}{Pattern 2}&\multicolumn{3}{c}{MCAR}\\
\cline{2-4}\cline{5-7}\cline{8-10}
	&
	Mean&Std. dev&Bias&Mean&Std. dev.&Bias&Mean&Std. dev.&Bias\\
	\hline
	$\vecmu_1$
	&$(0.65,-3.30)'$
	&$(1.20,1.22)'$
	&$(-0.35,-0.30)'$
	&$(0.53,-3.44)'$
	&$(1.23,1.08)'$
	&$(-0.47,-0.44)'$
	&$(0.39,-3.49)'$
	&$(1.48,1.53)'$
	&$(-0.61,-0.49)'$
	\\
	$\vecmu_2$
	&$(-0.58,3.27)'$
	&$(1.96,2.17)'$
	&$(0.42,0.27)'$
	&$(-0.00,4.09)'$
	&$(2.62,2.17)'$
	&$(1,1.09)'$
	&$(-0.81,3.25)'$
	&$(1.89,1.95)'$
	&$(0.19,0.25)'$
	\\
	$\vecbeta_1$
	&$(1.72,1.65)'$
	&$(1.52,1.55)'$
	&$(0.72,0.65)'$
	&$(1.82,1.78)'$
	&$(1.49,1.35)'$
	&$(0.82,0.78)'$
	&$(1.95,1.91)'$
	&$(1.81,1.79)'$
	&$(0.95,0.91)'$
	\\
	
	$\vecbeta_2$
	&$(-1.89,-1.72)'$
	&$(2.35,2.57)'$
	&$(-0.89,-0.72)'$
	&$(-2.63,-2.75)'$
	&$(3.10,2.56)'$
	&$(-1.63,-1.75)'$
	&$(-1.62,-1.78)'$
	&$(2.26,2.22)'$
	&$(-0.62,-0.78)'$
	\\
	$\vecmu_1+\vecbeta_1$
	&$(2.37,-1.65)'$
	&$(0.41,0.42)'$
	&$(0.37,0.35)'$
	&$(2.36,-1.65)'$
	&$(0.32,0.32)'$
	&$(0.36,0.35)'$
	&$(2.35,-1.58)'$
	&$(0.52,0.47)'$
	&$(0.35,0.42)'$
	\\
	$\vecmu_2+\vecbeta_2$
	&$(-2.47,1.55)'$
	&$(0.45,0.45)'$
	&$(-0.47,-0.45)'$
	&$(-2.63,1.34)'$
	&$(0.51,0.43)'$
	&$(-0.65,-0.66)'$
	&$(-2.42,1.47)'$
	&$(0.59,0.50)'$
	&$(-0.42,-0.53)'$
	\\
	$\matsig_1$
		&$\left[\begin{array}{cc}2.00&1.60\\1.60&2.00\end{array}\right]$
		&$\left[\begin{array}{cc}0.41&0.35\\0.35&0.38\end{array}\right]$
		&$\left[\begin{array}{cc}0.33&0.27\\0.27&0.33\end{array}\right]$
		&$\left[\begin{array}{cc}1.90&1.51\\1.51&1.90\end{array}\right]$
		&$\left[\begin{array}{cc}0.34&0.28\\0.28&0.33\end{array}\right]$
		&$\left[\begin{array}{cc}0.23&0.18\\0.18&0.23\end{array}\right]$
		&$\left[\begin{array}{cc}2.00&1.60\\1.60&1.98\end{array}\right]$			
		&$\left[\begin{array}{cc}0.56&0.48\\0.48&0.53\end{array}\right]$
		&$\left[\begin{array}{cc}0.33&0.27\\0.27&0.31\end{array}\right]$
		\\
	$\matsig_2$
		&$\left[\begin{array}{cc}4.44&3.56\\3.56&4.45\end{array}\right]$
		&$\left[\begin{array}{cc}0.79&0.70\\0.70&0.76\end{array}\right]$
		&$\left[\begin{array}{cc}1.11&0.89\\0.89&1.12\end{array}\right]$
		&$\left[\begin{array}{cc}4.33&3.45\\3.45&4.33\end{array}\right]$
		&$\left[\begin{array}{cc}0.76&0.66\\0.66&0.76\end{array}\right]$
		&$\left[\begin{array}{cc}1.00&0.78\\0.78&1.00\end{array}\right]$
		&$\left[\begin{array}{cc}4.37&3.49\\3.49&4.32\end{array}\right]$
		&$\left[\begin{array}{cc}0.96&0.85\\0.85&0.88\end{array}\right]$
		&$\left[\begin{array}{cc}1.04&0.82\\0.82&0.99\end{array}\right]$
		\\
	$\lambda_1$
	&$-2.26$
	&$1.34$
	&$-1.76$
	&$-2.73$
	&$1.70$
	&$-2.23$
	&$-2.61$
	&$1.90$
	&$-2.11$
	\\
	$\lambda_2$
	&$2.83$
	&$1.57$
	&$1.83$
	&$2.74$
	&$1.40$
	&$1.74$
	&$2.37$
	&$1.72$
	&$1.37$
	\\
	$\pi_1$&$0.50$
	&$0.01$
	&$0.00$
	&$0.50$
	&$0.01$
	&$0.00$
	&$0.50$
	&$0.01$
	&$0.00$
	\\
	$\pi_2$&$0.50$
	&$0.01$
	&$0.00$
	&$0.50$
	&$0.01$
	&$0.00$
	&$0.50$
	&$0.01$
	&$0.00$
	\\
	ARI&$0.90$
	&$0.03$
	&
	&$0.90$
	&$0.03$
	&
	&$0.90$
	&$0.09$
	&
	\\
	\hline	
	\end{tabular*}}
\end{sidewaystable}

\begin{sidewaystable}[ht]
\caption{Key model parameters as well as means, standard deviations and bias of the associated parameter estimations from the 100 runs for the first simulation experiment.}
\label{sim3parresult}
	\centering
	\tiny{
	\begin{tabular*}{1.1\textwidth}{@{\extracolsep{\fill}}cccccccccc}
		\hline
		\multicolumn{10}{c}{Sim3 using MST}\\
		\hline
		\multicolumn{10}{c}{$r=0.05$}\\
		&\multicolumn{3}{c}{Pattern 1}&\multicolumn{3}{c}{Pattern 2}&\multicolumn{3}{c}{MCAR}\\
\cline{2-4}\cline{5-7}\cline{8-10}
	&
	Mean&Std. dev&Bias&Mean&Std. dev.&Bias&Mean&Std. dev.&Bias\\
	\hline
	$\vecmu_1$
	&$(1.05,-3.18)'$
	&$(0.47,0.36)'$
	&$(0.05,-0.18)'$
	
	&$(0.98,-3.12)'$
	&$(0.48,0.35)'$
	&$(-0.02,-0.12)'$
	
	&$(1.00,-3.14)'$
	&$(0.47,0.36)'$
	&$(0.00,-0.14)'$
	\\
	$\vecmu_2$
	&$(-0.80,3.13)'$
	&$(0.58,0.38)'$
	&$(0.20,0.13)'$
	
	&$(-0.78,3.22)'$
	&$(0.58,0.40)'$
	&$(0.22,0.22)'$
	&$(-0.77,3.21)'$
	&$(0.58,0.43)'$
	&$(0.23,0.21)'$
	\\
	$\vecbeta_1$
	&$(0.93,1.25)'$
	&$(0.48,0.39)'$
	&$(-0.07,0.25)'$
	
	&$(0.95,1.23)'$
	&$(0.45,0.35)'$
	&$(-0.05,0.23)'$
	
	&$(0.96,1.20)'$
	&$(0.45,0.35)'$
	&$(-0.04,0.20)'$
	\\
	
	$\vecbeta_2$
	&$(-1.06,-1.18)'$
	&$(0.56,0.38)'$
	&$(-0.06,-0.18)'$
	
	&$(-1.16,-1.19)'$
	&$(0.55,0.42)'$
	&$(-0.16,-0.19)'$
	&$(-1.13,-1.26)'$
	&$(0.60,0.47)'$
	&$(-0.13,-0.26)'$
	\\
	$\vecmu_1+\vecbeta_1$
	&$(1.98,-1.93)'$
	&$(0.16,0.08)'$
	&$(-0.02,0.07)'$
	
	&$(1.93,-1.89)'$
	&$(0.12,0.07)'$
	&$(-0.07,0.11)'$
	
	&$(1.96,-1.93)'$
	&$(0.12,0.08)'$
	&$(-0.04,0.07)'$
	\\
	$\vecmu_2+\vecbeta_2$
	&$(-1.87,1.94)'$
	&$(0.23,0.10)'$
	&$(0.13,-0.06)'$
	
	&$(-1.94,2.03)'$
	&$(0.22,0.10)'$
	&$(0.06,0.03)'$
	
	&$(-1.90,1.95)'$
	&$(0.26,0.11)'$
	&$(0.10,-0.05)'$
	\\
	$\matsig_1$
		&$\left[\begin{array}{cc}3.36&0\\0&0.34\end{array}\right]$
		&$\left[\begin{array}{cc}0.46&0\\0&0.08\end{array}\right]$
		&$\left[\begin{array}{cc}0.36&0\\0&0.01\end{array}\right]$
		
		&$\left[\begin{array}{cc}3.33&0\\0&0.34\end{array}\right]$
		&$\left[\begin{array}{cc}0.42&0\\0&0.08\end{array}\right]$
		&$\left[\begin{array}{cc}0.33&0\\0&0.01\end{array}\right]$
		
		&$\left[\begin{array}{cc}3.32&0\\0&0.35\end{array}\right]$		
		&$\left[\begin{array}{cc}0.52&0\\0&0.09\end{array}\right]$
		&$\left[\begin{array}{cc}0.32&0\\0&0.02\end{array}\right]$
		\\
	$\matsig_2$
		&$\left[\begin{array}{cc}6.57&0\\0&0.66\end{array}\right]$
		&$\left[\begin{array}{cc}1.02&0\\0&0.14\end{array}\right]$
		&$\left[\begin{array}{cc}0.57&0\\0&-0.01\end{array}\right]$
		
		&$\left[\begin{array}{cc}6.58&0\\0&0.67\end{array}\right]$
		&$\left[\begin{array}{cc}1.16&0\\0&0.15\end{array}\right]$
		&$\left[\begin{array}{cc}0.58&0\\0&0.00\end{array}\right]$
		
		&$\left[\begin{array}{cc}6.51&0\\0&0.67\end{array}\right]$
		&$\left[\begin{array}{cc}1.14&0\\0&0.15\end{array}\right]$
		&$\left[\begin{array}{cc}0.51&0\\0&0.00\end{array}\right]$
		\\
	$\nu_1$
	&$8.25$
	&$3.18$
	&$1.25$
	
	&$8.14$
	&$3.01$
	&$1.14$
	
	&$8.00$
	&$2.77$
	&$1.00$
	\\
	$\nu_2$
	&$5.89$
	&$1.98$
	&$0.89$
	
	&$5.79$
	&$1.40$
	&$0.79$
	
	&$6.35$
	&$2.48$
	&$1.35$
	\\
	$\pi_1$&$0.50$
	&$0.02$
	&$0.00$
	&$0.50$
	&$0.01$
	&$0.00$
	&$0.50$
	&$0.02$
	&$0.00$
	\\
	$\pi_2$&$0.50$
	&$0.02$
	&$0.00$
	&$0.50$
	&$0.01$
	&$0.00$
	&$0.50$
	&$0.02$
	&$0.00$
	\\
	ARI&$0.81$
	&$0.03$
	&
	&$0.96$
	&$0.02$
	&
	&$0.81$
	&$0.09$
	&
	\\
	\hline	
		\multicolumn{10}{c}{$r=0.15$}\\
		&\multicolumn{3}{c}{Pattern 1}&\multicolumn{3}{c}{Pattern 2}&\multicolumn{3}{c}{MCAR}\\
\cline{2-4}\cline{5-7}\cline{8-10}
	&
	Mean&Std. dev&Bias&Mean&Std. dev.&Bias&Mean&Std. dev.&Bias\\
	\hline
	$\vecmu_1$
	&$(0.98,-3.26)'$
	&$(0.65,0.43)'$
	&$(-0.02,-0.26)'$
	
	&$(0.96,-3.26)'$
	&$(0.59,0.43)'$
	&$(-0.04,-0.26)'$
	
	&$(0.84,-3.26)'$
	&$(0.90,0.52)'$
	&$(-0.16,-0.26)'$
	\\
	$\vecmu_2$
	&$(-0.76,3.22)'$
	&$(0.50,0.42)'$
	&$(0.24,0.22)'$
	
	&$(-0.89,3.21)'$
	&$(0.47,0.37)'$
	&$(0.11,0.21)'$
	
	&$(-0.73,3.25)'$
	&$(0.64,0.60)'$
	&$(0.27,0.25)'$
	\\
	$\vecbeta_1$
	&$(0.99,1.33)'$
	&$(0.66,0.45)'$
	&$(-0.01,0.33)'$
	
	&$(0.98,1.32)'$
	&$(0.55,0.44)'$
	&$(-0.02,0.32)'$
	
	&$(1.08,1.34)'$
	&$(0.89,0.60)'$
	&$(0.08,0.34)'$
	\\
	
	$\vecbeta_2$
	&$(-1.09,-1.28)'$
	&$(0.49,0.44)'$
	&$(-0.09,-0.28)'$
	
	&$(-1.02,-1.27)'$
	&$(0.48,0.38)'$
	&$(-0.02,-0.27)'$
	
	&$(-1.12,-1.29)'$
	&$(0.46,0.52)'$
	&$(-0.12,-0.29)'$
	\\
	$\vecmu_1+\vecbeta_1$
	&$(1.98,-1.93)'$
	&$(0.14,0.08)'$
	&$(-0.02,0.07)'$
	
	&$(1.94,-1.94)'$
	&$(0.15,0.08)'$
	&$(-0.06,0.06)'$
	
	&$(1.92,-1.92)'$
	&$(0.25,0.20)'$
	&$(-0.08,0.08)'$
	\\
	$\vecmu_2+\vecbeta_2$
	&$(-1.86,1.94)'$
	&$(0.22,0.10)'$
	&$(0.14,-0.06)'$
	
	&$(-1.91,1.93)'$
	&$(0.22,0.11)'$
	&$(0.09,0.07)'$
	
	&$(-1.86,1.95)'$
	&$(0.32,0.22)'$
	&$(0.14,-0.05)'$
	\\
	$\matsig_1$
		&$\left[\begin{array}{cc}3.35&0\\0&0.32\end{array}\right]$
		&$\left[\begin{array}{cc}0.51&0\\0&0.08\end{array}\right]$
		&$\left[\begin{array}{cc}0.35&0\\0&-0.01\end{array}\right]$
		
		&$\left[\begin{array}{cc}3.38&0\\0&0.33\end{array}\right]$
		&$\left[\begin{array}{cc}0.55&0\\0&0.01\end{array}\right]$
		&$\left[\begin{array}{cc}0.38&0\\0&0.00\end{array}\right]$
		
		&$\left[\begin{array}{cc}3.36&0\\0&0.33\end{array}\right]$			
		&$\left[\begin{array}{cc}0.61&0\\0&0.09\end{array}\right]$
		&$\left[\begin{array}{cc}0.36&0\\0&0.00\end{array}\right]$
		\\
	$\matsig_2$
		&$\left[\begin{array}{cc}6.84&0\\0&0.65\end{array}\right]$
		&$\left[\begin{array}{cc}1.26&0\\0&0.15\end{array}\right]$
		&$\left[\begin{array}{cc}0.84&0\\0&-0.02\end{array}\right]$
		
		&$\left[\begin{array}{cc}6.79&0\\0&0.65\end{array}\right]$
		&$\left[\begin{array}{cc}1.13&0\\0&0.17\end{array}\right]$
		&$\left[\begin{array}{cc}0.79&0\\0&-0.02\end{array}\right]$
		
		&$\left[\begin{array}{cc}6.55&0\\0&0.63\end{array}\right]$
		&$\left[\begin{array}{cc}1.44&0\\0&0.17\end{array}\right]$
		&$\left[\begin{array}{cc}0.55&0\\0&-0.04\end{array}\right]$
		\\
	$\nu_1$
	&$8.65$
	&$4.21$
	&$1.65$
	
	&$9.27$
	&$5.90$
	&$2.27$
	
	&$9.37$
	&$6.46$
	&$2.37$
	\\
	$\nu_2$
	&$6.35$
	&$1.91$
	&$1.35$
	
	&$6.13$
	&$1.69$
	&$1.13$
	
	&$6.57$
	&$2.88$
	&$1.57$
	\\
	$\pi_1$&$0.50$
	&$0.01$
	&$0.00$
	&$0.50$
	&$0.01$
	&$0.00$
	&$0.50$
	&$0.02$
	&$0.00$
	\\
	$\pi_2$&$0.50$
	&$0.01$
	&$0.00$
	&$0.50$
	&$0.01$
	&$0.00$
	&$0.50$
	&$0.02$
	&$0.00$
	\\
	ARI&$0.78$
	&$0.04$
	&
	&$0.78$
	&$0.04$
	&
	&$0.74$
	&$0.04$
	&
	\\
	\hline		
			\multicolumn{10}{c}{$r=0.30$}\\
		&\multicolumn{3}{c}{Pattern 1}&\multicolumn{3}{c}{Pattern 2}&\multicolumn{3}{c}{MCAR}\\
\cline{2-4}\cline{5-7}\cline{8-10}
	&
	Mean&Std. dev&Bias&Mean&Std. dev.&Bias&Mean&Std. dev.&Bias\\
	\hline
	$\vecmu_1$
	&$(0.95,-3.30)'$
	&$(0.53,0.47)'$
	&$(-0.05,-0.30)'$
	
	&$(0.96,-3.20)'$
	&$(0.52,0.42)'$
	&$(-0.04,-0.20)'$
	
	&$(0.74,-3.21)'$
	&$(0.90,0.47)'$
	&$(-0.26,-0.21)'$
	\\
	$\vecmu_2$
	&$(-0.78,3.27)'$
	&$(0.58,0.41)'$
	&$(0.22,0.27)'$
	
	&$(-0.72,3.22)'$
	&$(0.60,0.38)'$
	&$(0.28,0.22)'$
	
	&$(-0.45,3.24)'$
	&$(1.18,0.51)'$
	&$(0.55,0.24)'$
	\\
	$\vecbeta_1$
	&$(1.05,1.27)'$
	&$(0.51,0.47)'$
	&$(0.05,0.27)'$
	
	&$(0.96,1.25)'$
	&$(0.50,0.43)'$
	&$(-0.04,0.25)'$
	
	&$(1.10,1.29)'$
	&$(0.93,0.49)'$
	&$(0.10,0.29)'$
	\\
	
	$\vecbeta_2$
	&$(-1.19,-1.27)'$
	&$(1.35,0.57)'$
	&$(-0.19,-0.27)'$
	
	&$(-1.16,-1.30)'$
	&$(0.55,0.43)'$
	&$(-0.16,-0.30)'$
	
	&$(-1.31,-1.32)'$
	&$(1.18,0.54)'$
	&$(-0.31,-0.32)'$
	\\
	$\vecmu_1+\vecbeta_1$
	&$(2.00,-2.03)'$
	&$(0.12,0.09)'$
	&$(0.00,-0.03)'$
	
	&$(1.92,-1.94)'$
	&$(0.17,0.09)'$
	&$(-0.08,0.06)'$
	
	&$(1.84,-1.92)'$
	&$(0.19,0.10)'$
	&$(-0.16,0.08)'$
	\\
	$\vecmu_2+\vecbeta_2$
	&$(-1.97,2.00)'$
	&$(0.25,0.11)'$
	&$(0.03,0.00)'$
	
	&$(-1.87,1.92)'$
	&$(0.25,0.13)'$
	&$(0.13,-0.08)'$
	
	&$(-1.76,1.92)'$
	&$(0.29,0.14)'$
	&$(0.24,-0.08)'$
	\\
	$\matsig_1$
		&$\left[\begin{array}{cc}3.26&0\\0&0.33\end{array}\right]$
		&$\left[\begin{array}{cc}0.54&0\\0&0.13\end{array}\right]$
		&$\left[\begin{array}{cc}0.26&0\\0&0.00\end{array}\right]$
		
		&$\left[\begin{array}{cc}3.26&0\\0&0.33\end{array}\right]$
		&$\left[\begin{array}{cc}0.52&0\\0&0.10\end{array}\right]$
		&$\left[\begin{array}{cc}0.26&0\\0&0.00\end{array}\right]$
		
		&$\left[\begin{array}{cc}3.24&0\\0&0.36\end{array}\right]$		
		&$\left[\begin{array}{cc}0.61&0\\0&0.14\end{array}\right]$
		&$\left[\begin{array}{cc}0.24&0\\0&0.03\end{array}\right]$
		\\
	$\matsig_2$
		&$\left[\begin{array}{cc}6.53&0\\0&0.67\end{array}\right]$
		&$\left[\begin{array}{cc}1.53&0\\0&0.20\end{array}\right]$
		&$\left[\begin{array}{cc}0.53&0\\0&0.00\end{array}\right]$
		
		&$\left[\begin{array}{cc}6.65&0\\0&0.67\end{array}\right]$
		&$\left[\begin{array}{cc}1.23&0\\0&0.18\end{array}\right]$
		&$\left[\begin{array}{cc}0.65&0\\0&0.00\end{array}\right]$
		
		&$\left[\begin{array}{cc}6.35&0\\0&0.67\end{array}\right]$
		&$\left[\begin{array}{cc}1.63&0\\0&0.20\end{array}\right]$
		&$\left[\begin{array}{cc}0.35&0\\0&0.00\end{array}\right]$
		\\
	$\nu_1$
	&$8.56$
	&$4.12$
	&$1.56$
	
	&$8.42$
	&$3.42$
	&$1.42$
	
	&$9.19$
	&$4.71$
	&$2.19$
	\\
	$\nu_2$
	&$6.83$
	&$2.57$
	&$1.83$
	
	&$6.34$
	&$2.02$
	&$1.34$
	
	&$7.26$
	&$4.37$
	&$2.26$
	\\
	$\pi_1$&$0.50$
	&$0.01$
	&$0.00$
	&$0.50$
	&$0.02$
	&$0.00$
	&$0.50$
	&$0.02$
	&$0.00$
	\\
	$\pi_2$&$0.50$
	&$0.01$
	&$0.00$
	&$0.50$
	&$0.02$
	&$0.00$
	&$0.50$
	&$0.02$
	&$0.00$
	\\
	ARI&$0.75$
	&$0.05$
	&
	&$0.76$
	&$0.05$
	&
	&$0.75$
	&$0.05$
	&
	\\
	\hline	
	\end{tabular*}}
\end{sidewaystable}

\begin{table}[htbp]
	\centering
	\caption{A comparsion of average BIC and ARI between MGHD, MST, and M\textit{t} models (replications=100) with $G=1,\ldots,4$.}
	
		\begin{tabular*}{1\textwidth}{@{\extracolsep{\fill}}lccccccc}
		  \hline
 		&&\multicolumn{2}{c}{MGHD} & \multicolumn{2}{c}{MST} & \multicolumn{2}{c}{M\textit{t}} \\ 
		      \cline{3-4}\cline{5-6}\cline{7-8}
		               &&BIC&ARI&BIC&ARI&BIC&ARI\\
 		 \hline
\multirow{3}{*}{Sim1}&$r=0.05$&$-1534$&0.95&$-1644$&0.88&$-1663$&0.75\\
&$r=0.15$&$-1412$&0.87&$-1517$&0.82&$-1559$&0.69\\
&$r=0.30$&$-1230$&0.74&$-1301$&0.69&$-1396$&0.60\\
\hline
\multirow{3}{*}{Sim2}&$r=0.05$&$-1647$&0.73&$-1683$&0.64&$-1823$&0.59\\
&$r=0.15$&$-1435$&0.62&$-1538$&0.52&$-1677$&0.48\\
&$r=0.30$&$-1201$&0.46&$-1266$&0.36&$-1463$&0.36\\
\hline
\multirow{3}{*}{Sim3}&$r=0.05$&$-1667$&0.82&$-1689$&0.76&$-1789$&0.64\\
&$r=0.15$&$-1517$&0.76&$-1502$&0.66&$-1622$&0.63\\
&$r=0.30$&$-1203$&0.70&$-1264$&0.60&$-1410$&0.48\\
\hline
\multirow{3}{*}{Sim4}&$r=0.05$&$-1546$&0.72&$-1608$&0.41&$-1849$&0.33\\
&$r=0.15$&$-1333$&0.60&$-1440$&0.37&$-1727$&0.27\\
&$r=0.30$&$-1142$&0.12&$-1171$&0.23&$-1385$&0.20\\
\hline
\multirow{3}{*}{Sim5}&$r=0.05$&$-1507$&0.94&$-1613$&0.74&$-1619$&0.88\\
&$r=0.15$&$-1366$&0.85&$-1507$&0.66&$-1450$&0.78\\
&$r=0.30$&$-1193$&0.71&$-1340$&0.59&$-1247$&0.64\\
\hline
\multirow{3}{*}{Sim6}&$r=0.05$&$-1356$&0.68&$-1445$&0.40&$-1614$&0.38\\
&$r=0.15$&$-1262$&0.58&$-1389$&0.38&$-1522$&0.35\\
&$r=0.30$&$-1130$&0.40&$-1263$&0.28&$-1385$&0.29\\

\hline
	\end{tabular*}
	\label{table:SimNG}
\end{table}

\end{document}